\documentclass[a4paper,12pt]{article}
\usepackage[utf8]{inputenc}
\usepackage{fancybox, calc}
\usepackage{authblk}
\usepackage{amssymb}
\usepackage{amsmath}
\usepackage{amsthm}
\usepackage{hyperref}
\usepackage{graphicx,wrapfig}
\usepackage{ucs}
\usepackage{cancel}
\usepackage{amsfonts}
\usepackage{amssymb}
\usepackage{makeidx}
\usepackage{diagbox}
\usepackage{tabularx}
\usepackage{float}
\usepackage[usenames, dvipsnames]{color}
\usepackage{url}
\usepackage{cite}
\usepackage{cancel}
\usepackage{transparent}
\usepackage{enumerate}
\usepackage{eso-pic}
\usepackage{subcaption}
\usepackage[normalem]{ulem}

\def\Rset{\mathbb{R}}
\def\pdf{f}

\newcommand{\down}[1]{\mathfrak {D}_{#1}}
\newcommand{\adown}{{\pdf}^\downarrow_{\alpha}}
\newcommand{\up}[1]{\mathfrak {U}_{#1}}
\newcommand{\aup}{{\pdf}^\uparrow_{\alpha}}

\newcommand{\pdfr}[0]{h}
\newcommand{\Appr}[2]{f^{({\mathfrak A,#1})}_{#2}}
\newcommand{\K}[3]{K_{#3}[#1||#2]}

\newcommand{\sign}[0]{\text{sign}}
\newtheorem{theorem}{Theorem}[section]
\newtheorem{corollary}{Corollary}[section]

\newtheorem{lemma}{Lemma}[section]
\newtheorem{definition}{Definition}[section]
\newtheorem{proposition}{Proposition}[section]

\newtheorem{remark}{Remark}[section]
\numberwithin{equation}{section}

\restylefloat{table}

\title{A new group of transformations related to the Kullback-Leibler and Rényi divergences and universal classes of monotone measures of statistical complexity}
\author[1]{Razvan Gabriel Iagar}
\affil[1]{Departamento de Matemática Aplicada, Ciencia e Ingeniería de los Materiales y Tecnología Electrónica, Universidad Rey Juan Carlos,
		28933 Móstoles (Madrid), Spain}
\author[1,2]{David Puertas-Centeno}
\affil[2]{Data, Complex Networks and Cybersecurity Research Institute, Universidad Rey Juan Carlos, 28028 (Madrid), Spain}
\author[1]{Elio V. Toranzo}
\date{\today}

\begin{document}
\maketitle

\begin{abstract}
	In this work we introduce a family of transformations, named \textit{divergence transformations}, interpolating between any pair of probability density functions sharing the same support. We prove the remarkable property that the whole family of Kullback-Leibler and Rényi divergences evolves in a monotone way with respect to the transformation parameter. Moreover, fixing the reference density, we show that the divergence transformations enjoy a group structure and can be derived through the algebraic conjugation of the recently introduced differential-escort transformations and their relative counterparts.

This algebraic structure allows us to deform any density function in such a way its divergence with respect a fixed reference density might also increase as much as possible. We also establish the monotonicity of composed measures involving the proper Kullback-Leibler and Rényi divergences as well as other recently introduced relative measures of moment and Fisher types.
		
As applications, an approximation scheme of general density functions by simple functions is provided. In addition, we give a number of analytical and numerical examples of interest in both regimes of increasing and decreasing divergence.
\end{abstract}

\section{Introduction}

The construction of operators allowing to transform or interpolate between different functions or function spaces, or to map new classes of functions into other classes whose properties are previously known, is a well-established path towards remarkable advances in mathematical physics, mathematical analysis, ordinary and partial differential equations and in general all the fields of nonlinear science. The availability of such transformations allows to improve our knowledge by either deriving new properties or extending already established properties to larger classes of general functions or solutions to differential equations. A number of such transformations became standard tools in mathematics and physics, such as the famous Fourier or Laplace transforms, but also the B\"acklund, Darboux or Sundman transformations, just to give a few examples.

In the framework of information theory, the so-called \textit{escort} transformation of a probability density (originally introduced in the context of nonextensive thermodynamics~\cite{Beck1995,Beck2004}) has been used to find bounds in source coding~\cite{Bercher2009}, but also to introduce new families of entropies and divergences~\cite{Bercher2011}, or establish new informational inequalities\cite{Bercher2012,Bercher2013}, among other applications~\cite{Bercher2012(divergences)}. 

Actually, certain families of transformations inspired by escort ones and linked with the Sundman transformations of power type, are recently playing a key role in the study of some relevant families of informational inequalities for probability density functions  (pdf, for short) with support on the real numbers~\cite{Puertas2025,IP2025,IP2025(b)}. Indeed, they have motivated the introduction of different generalized families of moment-type functionals of a pdf~\cite{Puertas2025,IP2025(b)} as well as functionals of entropic and informational nature~\cite{IP2025(b),IP2025(c)}. For the latter class of functionals, the Shannon's entropy \cite{Shannon1948} is considered a paradigmatic representative, from which numerous generalizations have been defined, such as, for example, the Rényi and Tsallis entropies \cite{Renyi1961,Tsallis1988,Tsallis2022}. These measures are strongly connected with the family of generalized Fisher information measures~\cite{Lutwak2005}, in view of the existence of inequalities of Stam, Cram\'er-Rao and moment-entropy type\cite{Lutwak2004,Lutwak2005,Zozor2015,Zozor2017} or equalities such as the De Brujin identity~\cite{Dembo1991,Toranzo2018}.

The aforementioned transformations, which have been named \textit{up and down}~\cite{IP2025} and \textit{differential-escort}~\cite{Puertas2019}, have been essential in highlighting an intricate structure involving the above mentioned families of functionals and inequalities~\cite{IP2025(b),IP2025(c)}. As a first fact, a mirrored domain of parameters in which the inequalities are satisfied is discovered, in which both the Stam-like and the moment-entropy inequalities exist, but not the Cramér-Rao one~\cite{IP2025}. Moreover, as a remarkable result, the new functionals induced by these transformations have allowed to establish new sharp upper bounds for the product of the classical Stam inequality, depending on regularity conditions~\cite{IP2025(b)}. Furthermore, the \textit{differential-escort} transformations and their algebraic conjugations with the \textit{up/down} ones provide groups of transformations such that the so-called LMC complexity measure \cite{LopezRuiz1995, LopezRuiz2005} and some generalizations of it~\cite{Puertas2019,IP2025(c)} are monotone with respect to these transformations in the sense of~\cite{Rudnicki2016}. Other interesting properties and applications to Hausdorff-like moment problems have been studied~\cite{Puertas2025,IP2025}, showing the usefulness of some of the newly introduced functionals to characterize heavy-tailed pdfs~\cite{Puertas2025}, in contrast to the escort mean values induced by the standard escort transformations~\cite{Tsallis2009}.

Moving to the framework of divergence measures, on which this work focuses, some of the more representative functionals are
the well-known Kullback-Leibler and Rényi divergences, as well as the Jensen-Shannon divergences, see for example \cite{Kullback1951,Lin2002,Martin2015,Taneja1989,Taneja2004,Yamano2009}. Likewise, regarding information-based functionals (such as the Fisher information \cite{Kharazmi2023}), there exist different relative versions \cite{Antolin2009,Martin2013,Toscani2017}, including Jensen-Fisher-based divergences whose structure resembles the Jensen-Shannon divergence \cite{Sanchez-Moreno2012,Yamano2021}. Many mathematical properties and applications of these divergence measures have been studied in the last decades~\cite{Harremoes2014}, including the obtention of new informational inequalities~\cite{Gilardoni2010,Guntuboyina2013,Sason2015,Sason2016,Zozor2021,Sason2022}. In this research line, the application of the techniques mentioned in the previous paragraph has been recently extended by the authors to the relative framework~\cite{IPT2025} to obtain new sharp inequalities relating the Kullback-Leibler and Rényi divergences to some new relative scaling-invariant functionals of moment and Fisher types. In order to prove these inequalities, the authors introduced a relative version of the differential-escort transformations, which in turn motivated the definition of the corresponding new relative functionals. We want to emphasize the relevant role of the scale invariance, which has not been taken so often into account in the literature regarding these type of divergences. However, this property is essential in order to be able to derive sharp bounds for inequalities in which such measures could be involved.

In this work, we continue with this research line by focusing on the monotonicity properties of the whole family of Kullback-Leibler and Rényi divergences. We thus define a new family of transformations interpolating between any pair of probability densities with a shared support, in such a way that the Kullback-Leibler and Rényi divergences are monotone with respect to the transformation parameter. We observe that these transformations are precisely the algebraic conjugations of differential-escort transformations with their relative counterparts, and consequently the group structure is naturally inherited. Furthermore, we introduce three new families of relative LMC complexity-like measures involving, respectively, the divergences and the moment and Fisher relative measures mentioned in the last paragraph. We study their basic properties and their monotone behavior with respect to either the divergence transformations or their algebraic conjugation with the \textit{up/down} transformations.

We support these theoretical findings with a number of examples of explicit calculations of the transformed density for several pdfs frequently employed in the general theory: exponential, Gaussian, Beta distribution, among others. Moreover, we compute the transformed density for a piecewise pdf and we illustrate these calculations by numerical experiments showing how a pdf either approaches or separates from a reference pdf as the Kullback-Leibler or the Rényi divergence vary in a monotone way.

We believe that all these results could be potentially interesting from an applied perspective as well. There is a growing interest in developing techniques to, on the one hand, be able to distinguish between different pdfs and, on the other hand, to increase this difference as much as possible in order to identify different patterns of behavior and extract information from them. This is a keypoint in the research meant to detect and predict anomalies in signals of different nature. Examples of these are temporal signals associated with biological systems, such as electroencephalograms (EEGs) or electrocardiograms (ECGs), and those generated by electrical systems such as gas turbines \cite{Squicciarini2024a,Squicciarini2024b}. Moreover, the standard LMC complexity has been widely used to analyze the structure of various types of electronic and molecular systems \cite{Puertas2019,Dehesa2023}. In fact, relative versions of the LMC complexity already exist, developed in the context of these types of systems, and some of them also satisfy the property of scaling invariance \cite{Nath2021,Borgoo2011}.

\bigskip

\noindent \textbf{Structure of the paper}. A number of preliminary definitions and results needed in the development of the paper are recalled to the reader, for the sake of completeness, in Section \ref{sec:prelim}. In Section \ref{sec:divertrans} we introduce the \textit{divergence transformations} and their inverse, we study their basic properties, as well as their relation with the Kullback-Leibler and Rényi divergences. Furthermore, we establish a monotonicity property of the Rényi divergence under this transformation. Section \ref{sec:stas_comp_mono} is dedicated to the definition of relative complexity measures based on the relative Fisher measures and relative cumulative moments. We therein explore the monotonicity property they fulfil (in the sense of~\cite{Rudnicki2016}), each of them involving a particular combination of the up/down and relative transformations. A large number of examples of the computation of the divergence transformations for some well-known pdfs are given in Section \ref{sec:examples}. Following up with examples, we explicitly compute in Section \ref{sec:Npiece} the divergence transformation of a $N$-piecewise pdf, $f_N$, with respect to an arbitrary reference pdf $h$, as well their Kullback-Leibler divergence.  As an application to the previous calculations, in Section \ref{sec:numerical} we illustrate through numerical experiments how this transformation works, in the sense of approaching or separating two pdfs, by selecting particular instances of $f$ and $h$ for each scenario. Finally, in Section \ref{sec:conclusions} we give some conclusions and state some open problems to be addressed in the future.

\section{Preliminary notions}\label{sec:prelim}

Throughout this paper, we consider as a general framework two probability density functions $f$ and $h$ such that
\begin{equation}\label{cond:support}
{\rm supp}\,f={\rm supp}\,h=\overline{\Omega}, \quad \Omega=(x_i,x_f)\subseteq\Rset, \quad f(x),  h(x)>0, \quad {\rm for \ any} \ x\in\Omega,
\end{equation}
where $\Omega$ can be either bounded or unbounded and $\overline{\Omega}$ denotes the closure of the set $\Omega$. We also set throughout the paper
\begin{equation}\label{eq:xi(alpha)}
\xi(\lambda,\alpha):=1+\alpha(\lambda-1), \quad \alpha, \ \lambda\in\Rset,
\end{equation}
and $\mathcal U=\frac{1}{x_f-x_i}$ designs the uniform density supported in $\Omega$.

\subsection{R\'enyi divergence. Basic definitions and properties}

We begin with a recap of the definitions of the R\'enyi and Kullback-Leibler divergences, representing the starting pieces of this research. Although the establisahed theory relies on considering $\xi>0$, in this work we extend the following definition to any $\xi\in\mathbb R$.
\begin{definition}[$\xi$-Rényi divergence]
Given a real number $\xi$ and two probability density functions $f$, $h$ satisfying \eqref{cond:support}, the $\xi$-Rényi divergence (or mutual information) is defined as
	\begin{equation}\label{eqdef:renyi_mutual}
	D_\xi[\pdf||h]=\frac1{\xi-1}\log\int_{\Omega} [\pdf(t)]^{\xi}[h(t)]^{1-\xi}\,dt,\quad \xi\neq 1.
	\end{equation}
In the limit $\xi\to1$ we recover the Kullback-Leibler divergence, which is defined as
	\begin{equation}\label{eqdef:kullback-leibler}
	D[f||h]:= \int_{\Omega}  f(t)\log\left[\frac{f(t)}{h(t)}\right]\,dt.
	\end{equation}
\end{definition}

\noindent Throughout the paper we will adopt the following notation
\begin{equation}\label{def:Renyiexp}
K_{\xi}[f||h]:=e^{(\xi-1)D_\xi[f||h]}=\int_{\Omega} [\pdf(t)]^{\xi}[h(t)]^{1-\xi}\,dt.
\end{equation}
Note that from the normalization condition of $\pdf$ and $\pdfr$, one trivially has $K_0[f||h]=K_{1}[f||h]=1$ for any pair of densities $f$, $h$. Moreover, given any pair of probability densities $f$ and $h$, the Rényi and Kullback-Leibler divergences are non-negative, that is,
	\begin{equation}\label{eq:pos_Renyi}
D_\xi[f||h]\geq0,
	\end{equation}
and the equality holds if and only if $f=h$. This implies that $K_{\xi}[f||h]>1$ for any $\xi>1$ and $K_{\xi}[f||h]\in(0,1)$ for any $\xi\in(0,1)$. The derivative of $K_{\xi}[f||h]$ with respect to $\xi$ is given by
\begin{equation}
\frac{d}{d\xi}K_{\xi}[f||h]=\int_{\Omega}[\pdf(t)]^{\xi}[h(t)]^{1-\xi}\log\left[\frac{f(t)}{h(t)}\right]\,dt,
\end{equation}
which in particular gives
\begin{equation}
\frac{d}{d\xi}K_{\xi}[f||h]\Big|_{\xi=0}=-D[h||f]\quad{\rm and}\quad \frac{d}{d\xi}K_{\xi}[f||h]\Big|_{\xi=1}=D[f||h].
\end{equation}
We end this preliminary section with the following convexity result.
\begin{remark}\label{remark:logK}
The function $\xi\mapsto\log K_{\xi}[f||h]$, or, equivalently, $\xi\mapsto(\xi-1)D_\xi[f||h],$ is convex in $\xi\in[0,\infty]$ (see ~\cite{Harremoes2014}), that is,
\begin{equation}
\frac{\partial^2}{\partial \xi^2} \log K_{\xi}[f||h]>0.
\end{equation}
\end{remark}

\subsection{Up and down transformations}\label{subsec:updown}

We next recall the recently introduced up/down transformations~\cite{IP2025, IP2025(b)}. They will be strongly employed in Section \ref{sec:stas_comp_mono} dedicated to measures of statistical complexity.

\medskip

\noindent \textbf{The down transformation.} The down transformation establishes a bijection between the set of decreasing and differentiable density functions and a more general class of density functions, as follows:
\begin{definition}\label{def:down}
Let $\pdf: \Omega\longrightarrow \Rset^+$ be a probability density function with $\Omega=(x_i,x_f),$ where $-\infty<x_i<x_f\le\infty$, such that $\pdf'(x)<0,\;\forall x\in\Omega$. Then, for $\alpha\in\Rset$, we define the transformation $\down{\alpha}[\pdf(x)]$ by
	\begin{equation}\label{eq:down}
	\pdf^{\downarrow}_\alpha(s)\equiv\down{\alpha}[\pdf(x)](s)=\pdf^\alpha(x(s))|\pdf'(x(s))|^{-1},\qquad s'(x)=\pdf^{1-\alpha}(x) |\pdf'(x)|.
	\end{equation}
\end{definition}
It is easy to see that $\down{\alpha}[\pdf]$ is a probability density function. Let us also observe that the class of transformations $\down{\alpha}$ is defined up to a translation, since $s(x)$ depends on an additive integration constant. Without loss of generality, the following \emph{canonical election} will be employed as a standard choice.
\begin{equation}\label{eq:can_change}
s(x)=\left\{\begin{array}{ll}\frac{\pdf^{2-\alpha}(x)}{\alpha-2}, & {\rm for} \ \alpha\in\Rset\setminus\{2\},\\
-\ln\,\pdf(x), & {\rm for} \ \alpha=2,\end{array}\right.
\end{equation}	
Many properties of the down transformation are established in previous works by the authors \cite{IP2025, IP2025(b)}.

\medskip

\noindent \textbf{The up transformation.} Contrary to the down transformation, the up transformation is applicable to any probability density function, and it returns a decreasing and differentiable probability density. We give below its precise definition.
\begin{definition}\label{def:up}
Let $\pdf: \Omega\longrightarrow \mathbb R^+$ be a probability density function with $\Omega=(x_i,x_f)$. For $\alpha\in\Rset\setminus\{2\}$, the up transformation $\up{\alpha}$ is defined as
\begin{equation}\label{eq:up}
	f^{\uparrow}_\alpha(u)=\up{\alpha}[\pdf(x)](u)=|(\alpha-2)x(u)|^\frac{1}{2-\alpha}, \quad u'(x)=-|(\alpha-2)x|^\frac{1}{\alpha-2}f(x),
\end{equation}
while for $\alpha=2$ the up transformation $\up{2}$ is defined as
\begin{equation}\label{eq:up2}
f^{\uparrow}_2(u)=\up{2}[\pdf(x)](u)=e^{-x(u)},\quad u'(x)=-e^xf(x).
\end{equation}
\end{definition}
We precise here that the definition of $u(x)$ in Eqs. \eqref{eq:up} and~\eqref{eq:up2} is taken up to a translation; for simplicity, the canonical choice is the primitive
$$
u(x)=\int_x^{x_f}|(\alpha-2)t|^\frac{1}{\alpha-2}f(t)dt, \quad \alpha\neq2,
$$
and the similar one for $\alpha=2$, where $x_f\in\Rset\cup\{\infty\}$ is the upper edge of the support of the domain of $f$, whenever this integral is finite. In the contrary case, we can employ any intermediate point $x_0$ in $(x_i,x_f)$ if the integral is divergent at its end.

The next result shows that the down and up transformations are mutually inverse.
\begin{proposition}\label{prop:inv}
Let $\pdf$ a probability density and let $\adown=\down{\alpha}[f]$ and $\aup=\up{\alpha}[f]$ be its $\alpha$-order down and up transformations. Then, up to a translation,
	\begin{equation}
	\pdf=\up{\alpha}[\adown]=\down{\alpha}[\aup].
	\end{equation}
that is, $\down{\alpha}\up{\alpha}=\up{\alpha}\down{\alpha}=\mathbb I$, where $\mathbb I$ denotes the identity operator (up to a translation).
\end{proposition}

\medskip

\noindent \textbf{Relation between up/down transformations, moments and entropies.} We next gather in the following technical result several identities showing how the up and down transformations relate to some of the classical informational functionals. Some of these properties will be employed thoughout the paper.
\begin{lemma}\label{lem:MEF}
Let $\pdf$ be a probability density and $\aup$ and $\adown$ its up/down transformations. Then, if $\alpha\in\Rset\setminus\{2\}$, the following equalities hold true:
	\begin{equation}\label{eq:ME}
	\sigma_p[\adown]=\frac{N_{1+(2-\alpha)p}^{\alpha-2}[\pdf]}{|2-\alpha|},\quad\text{or equivalently,}\quad N_{\lambda}[\aup]=\left(|2-\alpha| \sigma_{\frac{\lambda-1}{2-\alpha}}[\pdf]\right)^{\frac{1}{\alpha-2}},
	\end{equation}
and
	\begin{equation}\label{eq:EF}
	N_{\lambda}[\adown]=\phi_{1-\lambda,2-\alpha}^{2-\alpha}[\pdf],\quad\text{or equivalently,}\quad  \phi_{p,\beta}[f^{\uparrow}_{2-\beta}]=\left(N_{1-p}[\pdf]\right)^{\frac1{\beta}}.
	\end{equation}
For the Shannon entropy we have
	\begin{equation}\label{eq:SUD}
	S[\adown]=\alpha S[f]+\big \langle \log \left| f'\right|\big \rangle,\qquad S[\aup]=\frac1{2-\alpha}\left\langle \log |x|\right\rangle+\frac{\log|2-\alpha|}{2-\alpha}.
	\end{equation}
For $\alpha=2$, we have the following equalities:
\begin{equation}\label{eq:Sp_down2}
\sigma_p[\down{2}[\pdf]]=\left[\int_{\Rset}f(x)|\ln\,f(x)|^pdx\right]^{\frac{1}{p}},
\end{equation}
respectively
\begin{equation}\label{eq:Ren_down2}
N_\lambda[\down{2}[\pdf]]=\lim_{\widetilde \lambda\to0}\phi_{1-\lambda,\widetilde \lambda}^{\widetilde\lambda}[f]\equiv F_{1-\lambda,0}^\frac{1}{1-\lambda}[f].
\end{equation}
\end{lemma}

\subsection{Relative differential-escort transformations}

In the next definition, we recall the relative differential-escort transformation, introduced by the authors in their previous work \cite{IPT2025}. This transformation will be of great importance throughout the paper, as we can see our main divergence-transformation in an alternative way as the composition of suitable relative differential-escort transformations.
\begin{definition}\label{def:transf}
Let $f$ and $h$ be two probability density functions satisfying \eqref{cond:support} and $\alpha\in\Rset$. We define the \textit{relative differential-escort} transformed density of $\alpha$-order of $f$ as
\begin{equation}\label{eq:transf}
\mathfrak R_\alpha^{[\pdfr]}[\pdf](y):=\left(\frac{\pdf(x(y))}{\pdfr(x(y))}\right)^\alpha,\quad y'(x)=\pdf(x)^{1-\alpha}\pdfr(x)^{\alpha}.
\end{equation}
\end{definition}

\noindent For simplicity, we also employ the alternative notation
$$
f_\alpha^{[h]}(y)\equiv \mathfrak R_\alpha^{[\pdfr]}[\pdf](y).
$$
Note that for  $h\equiv 1$ the relative differential-escort transformation $\mathfrak R_\alpha^{[h]}$ reduces to the standard differential-escort one $\mathfrak E_\alpha$~\cite{Puertas2019}, which can be applied to any probability density $f$ independently of the condition~\eqref{cond:support}. The following connection between the relative differential-escort transformation and the (standard) differential-escort transformation, proved in \cite{IPT2025} is very useful:
\begin{equation}\label{eq:comp_escort}
\mathfrak E_{\beta}\mathfrak R_\alpha^{[h]}[f]=\mathfrak R_{\alpha\beta}^{[h]}[f].
\end{equation}
The inverse of the relative differential-escort transformation is a rather non-trivial construction and only works on a more restrictive class of density functions, satisfying the finiteness of a suitable R\'enyi entropy (which reduces to the compactness of the support in the limiting case $\alpha=1$). Employing the notation
$$
\mathcal{S}_a[f](x):=af(ax)
$$
for a scaling change, we have:
\begin{definition}\label{def:inverse}
Given a probability density $h$ as reference function, for any $\alpha\in\Rset\setminus\{1\}$ and any probability density function $g$ such that
\begin{equation}\label{cond:inverse}
r(\alpha,g):=N_{\frac{\alpha-1}{\alpha}}[g]=N_{\frac{1}{\alpha^*}}[g]<\infty,
\end{equation}
we define
\begin{equation}\label{eq:inversegen}
\mathfrak{R}^{-1,[h]}_{\alpha}[g](x)=\mathcal{S}_{r(\alpha,g)}[g](y(x))^{\frac{1}{\alpha}}h(x), \quad y'(x)=\mathcal{S}_{r(\alpha,g)}[g](y(x))^{-\frac{1}{\alpha^*}}h(x),
\end{equation}
which can be equivalently written as
\begin{equation}\label{eq:inversegen2}
\mathfrak{R}^{-1,[h]}_{\alpha}[g](x)=[rg(ry(x))]^{\frac{1}{\alpha}}h(x), \quad y'(x)=[rg(ry(x))]^{-\frac{1}{\alpha^*}}h(x),
\end{equation}
with $r=r(\alpha,g)$. For $\alpha=1$ and any compactly supported probability density function $g$, we define
\begin{equation}\label{eq:inverse}
\mathfrak{R}^{-1,[h]}_{1}[g](x)=\mathcal{S}_{r(1,g)}[g](y(x))h(x), \quad y'(x)=h(x),
\end{equation}
where $r(1,g)$ is the length of the support of $g$.
\end{definition}
The following result, established as \cite[Lemma 3.2]{IPT2025}, gives a connection between the R\'enyi entropy power and the relative differential-escort transformation.
\begin{lemma}\label{lem:Renyi}
Let $f$ and $h$ be two probability density functions satisfying \eqref{cond:support}, $\alpha\in\Rset$ and $\lambda\in\Rset\setminus\{1\}$. Then
	\begin{equation}\label{eq:Renyi}
	N_{\lambda}^{1-\lambda}[f_\alpha^{[h]}]=K_{\xi(\lambda,\alpha)}[f||h].
	\end{equation}
For the Shannon entropy, we find
	\begin{equation}\label{eq:Shannon}
	S[f_\alpha^{[h]}]=-\alpha D[f||h].
	\end{equation}
\end{lemma}

\subsection{Relative cumulative moment and relative Fisher divergence}

In this section, we recall two relative informational functionals introduced by the authors in their previous work \cite{IPT2025}, together with some identities connecting them to classical informational functionals. The first relative functional is the \emph{relative Fisher divergence}, defined below.
\begin{definition}\label{def:relativeFD}
Let $f$, $h$ be two probability density functions satisfying \eqref{cond:support}, being both derivable on their support, and let $(p,\lambda)\in\Rset^2$ be such that $p>1$ and $\lambda\neq0$. The \emph{relative Fisher divergence} is defined as
\begin{equation}\label{eq:relativeFD}
F_{p,\lambda}[f||h]:=\int_{\Omega}f^{1+p(\lambda-1)}h^{-\lambda p}\left|\frac d{dx}\left[\log\frac{f}{h}\right]\right|^p\,dx, \quad
\phi_{p,\lambda}[f||h]:=F_{p,\lambda}[f||h]^\frac{1}{p\lambda}
\end{equation}
\end{definition}
The following identity, connecting the relative Fisher divergence to the $(p,\lambda)$-Fisher information, will be useful in the section related to the monotonicity properties of some new measures of statistical complexity:
\begin{equation}\label{eq:relFI}
	\phi_{p,\lambda}[f_\alpha^{[h]}]=|\alpha|^{1/\lambda}\phi_{p,\lambda\alpha}[f||h]^{\alpha}.
\end{equation}
For a proof, the reader is referred to \cite[Lemma 4.1]{IPT2025}.

\medskip

The second informational functional recalled in this section is the \emph{relative cumulative moment}, representing the counterpart in the relative framework of the cumulative moments derived in \cite{Puertas2025}.
\begin{definition}\label{def:relativeCM}
Let $f$ and $h$ be two probability density functions satisfying \eqref{cond:support} and let $p>0$, $\alpha\in\Rset$. The \emph{relative cumulative moment} is defined as the expected value
$$
\mu_{p,\alpha}[f||h]:=\int_{\Omega}\left|\int_{x_i}^x f(s)^{1-\alpha}h(s)^\alpha\,ds\right|^pf(x)\,dx.
$$
We also introduce the quantity
$$
\sigma_{p,\alpha}[f||h]:=\mu_{p,\alpha}[f||h]^\frac1{p\alpha}.
$$
\end{definition}

\section{Divergence transformations: definitions and basic properties}\label{sec:divertrans}

In this section we introduce the main objects of study of the present paper, that is, the divergence transformations. Before giving their definition, we need a few preparatory notions. Let $\pdfr$ and $\pdf$ two probability densities satisfying \eqref{cond:support}. We introduce first the function $K(\cdot;x)$ given by
\begin{equation}\label{eq:Kx}
K(\alpha;x)=\int_{x_i}^x [\pdf(t)]^{\alpha}[\pdfr(t)]^{1-\alpha}\,dt.
\end{equation}
It follows obviously from \eqref{def:Renyiexp} that for $x=x_f$ we have $K(\alpha;x_f)=\K{\pdf}{\pdfr}{\alpha}$. Based on the function $K$, we next introduce a change of variable in the support $\Omega$. For a fixed $\alpha\in\Rset$, we define the function $y:\Omega\longrightarrow\Omega$ by
\begin{equation}\label{eq:y(x)}
y(x)=H^{-1}\left(\frac{K(1-\alpha;x)}{K_\alpha[\pdfr||\pdf]}\right),
\end{equation}
where $H^{-1}$ is the inverse function of the cumulative function of $\pdfr$,
\begin{equation}\label{eq:F0}
H(y)=\int_{x_i}^y \pdfr(s)ds, \quad y\in\Omega.
\end{equation}
The next easy result shows that, indeed, we have defined a change of variable inside $\Omega$.

\begin{lemma}\label{lemma:ybij}
The function $y:\Omega\longrightarrow\Omega$ defined in \eqref{eq:y(x)} is bijective.
\end{lemma}
\begin{proof}
On the one hand, taking into account the positivity of the reference density $h$ in $\Omega$, we deduce that both $K(\alpha;x)$ and $H$ (and thus also $H^{-1}$) are increasing functions with respect to $x$. It then follows from \eqref{eq:y(x)} that $y$ is an increasing function. On the other hand, we observe that
$$
y(x_i)=H^{-1}(0)=x_i \quad {\rm and} \quad y(x_f)=H^{-1}(1)=x_f,
$$
establishing that $y$ is onto and completing the proof.
\end{proof}
With the previous preparations in mind, we are now in a position to formulate the definition of our new family of transformations.
\begin{definition}[Divergence transformation]\label{defi:tilderho}
Let $\pdfr(x)$ and $\pdf(x)$ be two probability densities satisfying \eqref{cond:support} and let $\alpha\in\mathbb{R}$ be such that
\begin{equation}\label{cond:KL}
K_\alpha[\pdfr||\pdf]<\infty.
\end{equation}
We define the $\alpha$-divergence transformation $\mathfrak A^{(\pdfr)}_\alpha$ by
\begin{equation}\label{eq:tilderho}
\mathfrak{A}^{(\pdfr)}_\alpha[f](y)\equiv\pdf_{\alpha}^{(\mathfrak{A},\pdfr)}(y):= K_\alpha[\pdfr||\pdf]\,\left(\frac{ \pdf(x(y))}{\pdfr(x(y))}\right)^\alpha\; \pdfr(y),
\end{equation}
where $x(y)$ is the inverse function of the function defined in \eqref{eq:y(x)}.
\end{definition}
The first fact that we prove is that, indeed, by the transformation introduced in Definition \ref{defi:tilderho} we obtain probability density functions.
\begin{lemma} In the same conditions as in Definition \ref{defi:tilderho}, the $\alpha$-transformed function $\Appr{\pdfr}{\alpha}:\Omega \longrightarrow\mathbb R^+$ satisfies
\begin{equation}
\Appr{\pdfr}{\alpha}(y)dy=\pdf(x)dx.
\end{equation}
In particular, $\Appr{\pdfr}{\alpha}$ is a probability density with respect to the variable $y$.
\end{lemma}
\begin{proof}
On the one hand, we infer from \eqref{eq:F0} and \eqref{eq:y(x)} that
\begin{equation}\label{eq:y(x)0}
\int_{x_i}^y \pdfr(s)ds=\frac1{K_\alpha[\pdfr||\pdf]}\,\,\int_{x_i}^x [\pdf(t)]^{1-\alpha}\,[\pdfr(t)]^\alpha\,dt,
\end{equation}
which, by differentiating with respect to $x$ in both sides, leads to
\begin{equation}\label{eq:dem_lemma3_1}
\pdfr(y)\,y'(x)=\frac{[\pdf(x)]^{1-\alpha}[\pdfr(x)]^{\alpha}}{K_\alpha[\pdfr||\pdf]}.
\end{equation}
On the other hand, Eq.~\eqref{eq:tilderho} gives
\begin{equation}\label{eq:dem_lemma3_2}
\Appr{\pdfr}{\alpha}(y)dy=K_\alpha[\pdfr||\pdf]\,\left(\frac{ \pdf(x)}{\pdfr(x)}\right)^\alpha \pdfr(y)\,y'(x)dx.
\end{equation}
The conclusion follows by inserting \eqref{eq:dem_lemma3_1} into \eqref{eq:dem_lemma3_2}.
\end{proof}

\medskip

\noindent \textbf{Remark.} Define
\begin{equation}\label{eq:asup}
\alpha_c^{+}:= \sup\{ \alpha\in\mathbb{R}: K_\alpha[\pdfr||\pdf] <\infty \}, \quad \alpha_c^{-}:= \inf\{ \alpha\in\mathbb{R}: K_\alpha[\pdfr||\pdf] <\infty \}.
\end{equation}
Assuming $f,h$ bounded in their common support $\Omega$, we observe from \eqref{eq:y(x)} that, whenever $K_{\alpha_c^+}[h||f]=\infty$ or $K_{\alpha_c^-}[h||f]=\infty$, the transformed density approaches a Dirac distribution concentrated at $y =x_i$ as $\alpha\to\alpha_c^{+}$ and also as $\alpha\to\alpha_c^{-}$, since
\[
y(x) = H^{-1}(0) = x_i \, ,
\]
for any $x\in \Omega$. However, it is also possible that $K_{\alpha_c^+}[h||f]<\infty$ or $K_{\alpha_c^-}[h||f]<\infty$, depending on the densities $h$ and $f$, and such an example is given in Section \ref{sec:expGauss} below. Note that both critical values $\alpha_c^{+}$ and $\alpha_{c}^{-}$ can be finite or infinite.

\medskip

We give next a number of elementary properties of the family of transformations $\mathfrak A^{(\pdfr)}_\alpha$. The first property expresses the transformation $\mathfrak{A}_{\alpha}^{(h)}$ in terms of a composition of relative differential-escort transformations.
\begin{proposition}\label{prop:relative_escort}
For any probability density functions $f$ and $h$ satisfying \eqref{cond:support} and $\alpha\in\Rset$ such that \eqref{cond:KL} is fulfilled, we have
\begin{equation}\label{eq:appr_comp}
\mathfrak{A}^{(\pdfr)}_\alpha[f]=\mathfrak{R}_1^{-1,[h]}[\mathfrak{R}_{\alpha}^{[h]}[f]]=\left(\mathfrak{R}_1^{-1,[h]}\circ\mathfrak E_\alpha\circ\mathfrak{R}_{1}^{[h]}\right)[f],
\end{equation}
where the notations with $\mathfrak{R}$ indicate the relative differential-escort transformation and its inverse introduced in Definitions \ref{def:transf} and \ref{def:inverse}.
\end{proposition}
\begin{proof}
Let us first note that the condition \eqref{cond:KL} gives that
$$
\int_{x_i}^{x_f}h(t)^{\alpha}f(t)^{1-\alpha}\,dt<\infty,
$$
which implies that $\mathfrak{R}_{\alpha}^{[h]}[f]$ is a compactly supported density, according to Eq. \eqref{eq:transf}, and in fact the length of its support is given exactly by $K_{\alpha}[h||f]$. We thus deduce that the inverse $\mathfrak{R}_1^{-1,[h]}$ can be applied (according to \eqref{eq:inverse}) and the right-hand side of \eqref{eq:appr_comp} is well defined. In order to prove the equality, we start from the right-hand side and employ \eqref{eq:inverse} to find
\label{key}\begin{equation}\label{eq:interm7}
\begin{split}
\mathfrak{R}_1^{-1,[h]}[\mathfrak{R}_{\alpha}^{[h]}[f(x)]](z)&=\mathfrak{R}_1^{-1,[h]}[f_{\alpha}^{[h]}(y)](z)\\
&=S_{r(1,f_{\alpha}^{[h]})}[f_{\alpha}^{[h]}](y(z))h(z)\\
&=r(1,f_{\alpha}^{[h]})f_{\alpha}^{[h]}\bigg(r(1,f_{\alpha}^{[h]})y(z)\bigg)h(z).
\end{split}
\end{equation}
Since $r(1,g)$ designs the length of the support of the probability density function $g$, it follows that
\begin{equation}\label{eq:interm8}
r(1,f_{\alpha}^{[h]})=\int_{x_i}^{x_f}h(t)^{\alpha}f(t)^{1-\alpha}\,dt=K_{\alpha}[h||f].
\end{equation}
Moreover, we deduce from Definition \ref{def:transf} that
\begin{equation}\label{eq:interm9}
f_{\alpha}^{[h]}(ry(z))=\left(\frac{f(x(ry(z)))}{h(x(ry(z)))}\right)^{\alpha}, \quad r=r(1,f_{\alpha}^{[h]}).
\end{equation}
Since it is easy to see that $x(ry(z))=x(z)$ by conservation of the measure, by inserting \eqref{eq:interm8} and \eqref{eq:interm9} into \eqref{eq:interm7} we obtain that
$$
\mathfrak{R}_1^{-1,[h]}[\mathfrak{R}_{\alpha}^{[h]}[f(x)]](z)=K_{\alpha}[h||f]\left(\frac{f(x(z))}{h(x(z))}\right)^{\alpha}h(z),
$$
thus obtaining exactly the formula in \eqref{eq:tilderho} in the right-hand side and completing the proof.
\end{proof}
The algebraic representation established in \eqref{eq:appr_comp} will be very useful in the sequel. The next property is related to the composition of two different divergence transformations and reveals a group structure.
\begin{proposition}[Composition and inverse element]\label{prop:group}
Let $\alpha$, $\beta\in\Rset$. Then we have
\begin{equation}\label{eq:group}
\mathfrak{A}_{\beta}^{(h)}\circ \mathfrak{A}_{\alpha}^{(h)}=\mathfrak{A}_{\beta\alpha}^{(h)}.
\end{equation}
In particular, for any $\alpha\neq0$, the $\alpha$-divergence transformation is invertible and
\begin{equation}\label{eq:inverse_appr}
\bigg(\mathfrak{A}_{\alpha}^{(h)}\bigg)^{-1}=\mathfrak{A}_{\frac{1}{\alpha}}^{(h)}.
\end{equation}
\end{proposition}
Before going to the proof, let us stress here that, if restricting ourselves to nonzero parameters, we have actually introduced a \emph{group of transformations}.
\begin{proof}
We compute the composition $\mathfrak A_\alpha^{(h)}\circ \mathfrak A_\beta^{(h)}$ by employing \eqref{eq:appr_comp} and \eqref{eq:comp_escort}. We thus have
\begin{equation*}
\begin{split}
\mathfrak{A}_{\alpha}^{(h)}\circ\mathfrak{A}_{\beta}^{(h)}&=\mathfrak{R}_{1}^{-1,[h]}\circ\mathfrak{R}_{\alpha}^{[h]}\mathfrak{R}_{1}^{-1,[h]}\circ\mathfrak{R}_{\beta}^{[h]}\\
&=\mathfrak{R}_{1}^{-1,[h]}\circ\mathfrak{E}_{\alpha}\mathfrak{R}_{1}^{[h]}\circ\mathfrak{R}_{1}^{-1,[h]}\circ\mathfrak{R}_{\beta}^{[h]}\\
&=\mathfrak{R}_{1}^{-1,[h]}\circ\mathfrak{R}_{\alpha\beta}^{[h]}=\mathfrak{A}_{\alpha\beta}^{(h)},
\end{split}
\end{equation*}
as claimed.
\end{proof}
The following property expresses the behavior of the divergence transformation with respect to a change of scale.
\begin{proposition}[Scaling changes]\label{prop:scaling}
Let $\alpha\in\mathbb{R}$, $a\in\mathbb{R}^{+}$ and $f$, $h$ be two probability density functions such that \eqref{cond:support} and \eqref{cond:KL} are fulfilled. Consider the scaled density functions,
$$
f_a(x):=af(ax) \quad \text{and} \quad h_a(x):=ah(ax)\, \quad x\in\Omega_a:=\left(\frac{x_i}{a},\frac{x_f}{a}\right).
$$
Then,
\begin{equation}\label{eq:scaling}
\mathfrak{A}_{\alpha}^{(h_a)}[f_a(x)](y) = \left(\mathfrak{A}_{\alpha}^{(h)}[f(x)](y)\right)_a\,, \quad a>0.
\end{equation}
\end{proposition}
\begin{proof}
We obtain from the definition \eqref{eq:tilderho} that
	\begin{align}\nonumber\label{eqproof0}
		\mathfrak A_{\alpha}^{(h_a)}[f_a(x)](y) & = f_{a,\alpha}^{(h_a)}(y) = K_\alpha[\pdfr_a|| f_a]\left[ \frac{f_a(x(y)) }{\pdfr_a(x(y)) }\right]^{\alpha} h_a(y) \\
		&= K_\alpha[\pdfr_a|| f_a]\left[ \frac{af(ax(y)) }{a\pdfr(ax(y)) }\right]^{\alpha} a\pdfr(ay),
\end{align}
where
the change of variable $y=y(x)$ satisfies
	\begin{equation}\label{eqproof}
	\int_{x_i/a}^{y} h_a(u)\, du=\frac{\int_{x_i/a}^{x} h_a^{\alpha}(t) f_a^{1-\alpha}(t)\, dt   }{K_\alpha[\pdfr_a||\pdf_a] }.
	\end{equation}
Computing the factors of the right hand side  of the  Eq.~\eqref{eqproof0}, we find on the one hand that
	\begin{align*}
		K_\alpha[\pdfr_a|| f_a] &= \int_{\Omega_a} h_a^{\alpha}(x)f_a^{1-\alpha}(x)\, dx = \int_{\Omega_a} a^{\alpha}h^{\alpha}(ax)a^{1-\alpha}f^{1-\alpha}(ax)\, dx  \\
		& =\int_{\Omega_a} h^{\alpha}(ax) f^{1-\alpha}(ax)\, d(ax) =  \int_{\Omega} h^{\alpha}(t) f^{1-\alpha}(t)\, d(t) = 	K_\alpha[\pdfr|| f]\,.
	\end{align*}
On the other hand, the change of variable in Eq.~\eqref{eqproof} can be written as
	\begin{equation}\label{eqproof1}
\int_{x_i}^{ay} h(\overline u)\, d\overline u = \frac{ K(1-\alpha;ax)}{K_\alpha[\pdfr||\pdf]} = \frac{\int_{x_i}^{ax}  h^{\alpha}( \overline t) f^{1-\alpha}(\overline t)\, d\overline t   }{K_\alpha[\pdfr||\pdf] }.
	\end{equation}
	Adopting the notation $\overline y=\overline y(\overline x)$ for the change of variable given by
	\[
	\int_{x_i}^{\overline y} h(\overline  u)\, d\overline u=\frac{\int_{x_i}^{\overline x} h^{\alpha}(\overline t) f^{1-\alpha}(\overline t)\, d\overline t   }{K_\alpha[\pdfr||\pdf] },
	\]
	 the Eq.~\eqref{eqproof1} implies that $a x(y)=\overline x(ay)$. It thus follows that
	\begin{align*}
	\mathfrak A_\alpha^{(h_a)}[f_a(x)](y) &= a K_\alpha[\pdfr|| f]\left[ \frac{f(ax(y)) }{\pdfr(ax(y)) }\right]^{\alpha} \pdfr(ay) = 	\ a K_\alpha[\pdfr|| f]\left[ \frac{f(\overline x(ay)) }{\pdfr(\overline x(ay)) }\right]^{\alpha} \pdfr(ay) \\
	& = a\, \mathfrak A_\alpha^{(h)}[f(\overline x)](ay) =  \left(\mathfrak A_\alpha^{(h)}[f(\overline x)](y)\right)_a= \left(\mathfrak A_\alpha^{(h)}[f](y)\right)_a,
	\end{align*}
as claimed.
\end{proof}

\noindent \textbf{Remark.} As we have seen in the previous results, the class of divergence transformations conserves both the support of the density and the measure in each interval. This is the role of $K_\alpha[\pdfr||\pdf]$; indeed, without its presence in \eqref{eq:y(x)}, the transformation would be well-defined only for $\alpha\in[0,1]$, since for $\alpha>1$ there is no $y_f\in\mathbb R$ such that
\begin{equation*}
\int_{x_i}^{y_f} \pdfr(s)ds=\int_{x_i}^{x_f} [\pdf(x)]^{1-\alpha}\,[\pdfr(t)]^\alpha\,dt=\K{f}{h}{\alpha}>1.
\end{equation*}
In the next result we calculate the R\'enyi and Kullback-Leibler divergences of a divergence-transformed density.
\begin{proposition}\label{prop:KL}
Let $\pdfr$ and $\pdf$ be two probability densities satisfying \eqref{cond:support} and let $\alpha$, $\xi\in\Rset$. Let $\Appr{\pdfr}{\alpha}$ be the $\alpha$-transformed density of $\pdf$. Then, the R\'enyi divergence $D_\xi[\Appr{\pdfr}{\alpha}||\pdfr]$ satisfies
\begin{equation}\label{eqlemma:appr_diver}
D_\xi[\Appr{\pdfr}{\alpha}||\pdfr]=\alpha\, D_{\xi_\alpha}[\pdf||\pdfr]+(\alpha-1)D_\alpha[h||f],\quad \xi_\alpha=1+\alpha(\xi-1).
\end{equation}
The latter equation can be written equivalently as
\begin{equation}\label{eqlemma:appr_diver2}
D_\xi[\Appr{\pdfr}{\alpha}||\pdfr]=\alpha\,(D_{\xi_\alpha}[\pdf||\pdfr]-D_{1-\alpha}[f||h]).
\end{equation}
\end{proposition}
\begin{proof}
The definition of the R\'enyi divergence \eqref{eqdef:renyi_mutual} ensures that
\begin{equation}\label{eq:interm1}
D_\xi[\Appr{\pdfr}{\alpha}||\pdfr]=\frac1{\xi-1}\log\int_{\Omega} \Appr{\pdfr}{\alpha}(y) \left(\frac{\Appr{\pdfr}{\alpha}(y)}{\pdfr(y)}\right)^{\xi-1}\,dy.
\end{equation}
The change of variable given in \eqref{eq:y(x)}, the definition of the divergence transformation in \eqref{eq:tilderho} and \eqref{eq:interm1} give, after adopting the notation $\xi_\alpha=1+\alpha(\xi-1)$, that
\begin{equation*}
\begin{split}
D_\xi[\Appr{\pdfr}{\alpha}||\pdfr]&=\frac1{\xi-1}\log\int_{\Omega} \pdf(x) \left(\K{h}{f}{\alpha}\left(\frac{\pdf(x)}{\pdfr(x)}\right)^\alpha\right)^{\xi-1}\,dx\\
&=\frac1{\xi-1}\log\int_{\Omega} \pdf(x) \left(\frac{\pdf(x)}{\pdfr(x)}\right)^{\alpha(\xi-1)}\,dx\nonumber+\log\K{h}{f}{\alpha}\\
&=\frac{1}{\xi-1}\log K_{\xi_\alpha}[\pdf||\pdfr]+\log \K{\pdfr}{\pdf}{\alpha}\\
&=\alpha\, D_{\xi_\alpha}[\pdf||\pdfr]+\log \K{\pdfr}{\pdf}{\alpha},
\end{split}
\end{equation*}
proving Eq.~\eqref{eqlemma:appr_diver}.
Finally, from the fact that $K_{1-\alpha}[f||h]=K_{\alpha}[h||f]$ trivially follows Eq.~\eqref{eqlemma:appr_diver2}, which concludes the proof.
\end{proof}

\noindent \textbf{Remark.} Note that Eq.~\eqref{eqlemma:appr_diver} leads to
	
	\[
	D_\beta\left[h_{1-\alpha}^{(\mathfrak A, f)}\big|\big| f\right]=(1-\alpha)D_{\beta_{1-\alpha}}[h||f]-\alpha D_{1-\alpha}[f||h]
	\]
	where $\beta_{1-\alpha} = \alpha + \beta(1-\alpha)$ and, taking into account the obvious identity following from Eq. \eqref{eqdef:renyi_mutual}
$$
(\alpha-1)D_\alpha[h||f]=-\alpha D_{1-\alpha}[f||h],
$$
we obtain

$$
D_\beta\left[h_{1-\alpha}^{(\mathfrak A,f)}\big|\big| f\right]=(1-\alpha)\frac{\beta_{1-\alpha} }{1- \beta_{1-\alpha} } D_{1-\beta_{1-\alpha}}[f||h]+(\alpha-1)D_\alpha[h||f].
$$
Choosing $\beta$ and $\xi$ such that $\xi_{\alpha} = 1-\beta_{1-\alpha}$, we obtain that
\[
\beta = \frac{\alpha \, \xi }{\alpha -1} \quad \text{and} \quad (1-\alpha)\frac{\beta_{1-\alpha} }{1- \beta_{1-\alpha} } = \alpha \left( 1 - \frac{ \xi }{ \xi_{\alpha } }\right)\, .
\]
Then, we infer from Eq.~\eqref{eqlemma:appr_diver} that

\begin{equation*}
	D_\xi\left[f_{\alpha}^{(\mathfrak A, h)}\big|\big| h\right]=D_{\frac{\alpha\,\xi}{\alpha-1}}\left[h_{1-\alpha}^{(\mathfrak A, f)}\big|\big|f \right] + \alpha\frac{\xi}{\xi_{\alpha}}D_{\xi_{\alpha}}[f||h].
\end{equation*}

We conclude this section with the following important monotonicity result, showing how the divergence transformed densities approach the reference function $h$ as $\alpha\to0$ or, on the contrary, depart from $h$ as $|\alpha|$ increases.
\begin{theorem}\label{th:divergence}
Under the conditions of Proposition~\ref{prop:KL} and for $\xi\geqslant1$, the function $D_\xi[\Appr{\pdfr}{\alpha}||\pdfr]$ is increasing with respect to the parameter $\alpha\geqslant0$ and decreasing with respect to $\alpha\leqslant0$. Moreover, for $\alpha=0$ one has $D_\xi[\Appr{\pdfr}{0}||\pdfr]=0$ or, equivalently, $\Appr{\pdfr}{0}=\pdfr.$
\end{theorem}
\begin{proof}
Letting $\alpha=0$ in Eq.~\eqref{eqlemma:appr_diver}, we deduce that $D_\xi[\Appr{\pdfr}{0}||\pdfr]=0$ and thus $\Appr{\pdfr}{0}=\pdfr.$ The positivity of the R\'enyi divergence stated in \eqref{eq:pos_Renyi} ensures that $\alpha=0$ is a minimum point for the application $\alpha\mapsto D_\xi[\Appr{\pdfr}{\alpha}||\pdfr]$ and thus
\begin{equation}\label{eq:interm6}
\frac{\partial}{\partial\alpha}D_\xi[\Appr{\pdfr}{\alpha}||\pdfr]\Big|_{\alpha=0}=0.
\end{equation}
In order to establish the monotonicity with respect to $\alpha$, we start from the equality
\begin{equation}\label{eq:interm4}
D_\xi[\Appr{\pdfr}{\alpha}||\pdfr]=\frac{1}{\xi-1}\log K_{\xi_\alpha}[\pdf||\pdfr]+\log \K{\pdfr}{\pdf}{\alpha},
\end{equation}
established in the proof of Proposition \ref{prop:KL}. We observe that the right-hand side of \eqref{eq:interm4} is differentiable with respect to $\alpha$, as it readily follows from \eqref{eqdef:renyi_mutual} and \eqref{def:Renyiexp}. By taking derivatives with respect to $\alpha$ up to the second order in \eqref{eq:interm4}, we find
\begin{equation}\label{eq:interm5}
\frac{\partial^2(D_\xi[\Appr{\pdfr}{\alpha}||\pdfr])}{\partial\alpha^2}=(\xi-1)\frac{\partial^2}{\partial\xi_\alpha^2}\left(\log K_{\xi_\alpha}[f||h]\right)+\frac{\partial^2}{\partial\alpha^2}\left(\log K_{\alpha}[h||f]\right)\geq0,
\end{equation}
where the positivity follows from the convexity of the function $\log K_\alpha[f||h]$ with respect to $\alpha$, as mentioned in Remark~\ref{remark:logK}. We infer from \eqref{eq:interm5} that the first derivative with respect to $\alpha$ of the mapping $\alpha\mapsto D_\xi[\Appr{\pdfr}{\alpha}||\pdfr]$ is non-negative for any $\alpha>0$ and non-positive for any $\alpha<0$. Combining this fact with \eqref{eq:interm6}, we deduce that
$$
\frac{\partial}{\partial\alpha}D_\xi[\Appr{\pdfr}{\alpha}||\pdfr]\geqslant0, \quad {\rm for \ any} \ \alpha>0, \quad
\frac{\partial}{\partial\alpha}D_\xi[\Appr{\pdfr}{\alpha}||\pdfr]\leqslant0, \quad {\rm for \ any} \ \alpha<0,
$$
which is equivalent to the claimed monotonicity with respect to $\alpha$, completing the proof.
\end{proof}
\begin{corollary} Given $\xi\geqslant1$, we have
\begin{equation}
D_\xi[\Appr{\pdfr}{\alpha}||\pdfr]>D_\xi[\pdf||\pdfr],\quad {\rm for} \ \alpha>1;\qquad D_\xi[\Appr{\pdfr}{\alpha}||\pdfr]<D_\xi[\pdf||\pdfr],\quad {\rm for} \ \alpha\in(0,1).
\end{equation}
\end{corollary}

\section{Measures of statistical complexity and monotonicity}\label{sec:stas_comp_mono}

As already commented in the Introduction, the so-called statistical complexity measures have been proposed as important tools in order to measure the degree of order and disorder of physical systems. This is because they satisfy a number of relevant mathematical properties such as the invariance under scaling trasnformations or the existence of a minimal bound, among others. An analysis of the most important complexity measures and the property of monotonicity can be found in \cite{Rudnicki2016}, one of the most celebrated being the LMC complexity measure (see \cite{LopezRuiz1995, LopezRuiz2005}), extended later to the LMC-R\'enyi complexity measure. A monotonicity property for the latter complexity measure, involving the differential-escort transformation, has been proved in \cite{Puertas2019}. Since the relative differential-escort transformation strongly generalizes the differential-escort one, it is natural to raise the question of proposing relative complexity measures and establishing monotonicity properties based on it. This is the goal of this section.

\medskip

\noindent \textbf{A relative LMC-R\'enyi complexity measure.} The first measure of complexity we propose is the following natural extension to the relative framework of the above mentioned LMC-R\'enyi complexity measure.
\begin{definition}\label{def:complex}
Let $f$ and $h$ be two probability density functions satisfying \eqref{cond:support} and let $(\lambda,\beta)\in\Rset^2$ such that $\lambda\neq1$, $\beta\neq1$. We define the \emph{relative LMC-R\'enyi complexity measure} by
\begin{equation}\label{eq:complex} C^{(D)}_{\lambda,\beta}[f||h]:=\frac{K_{\lambda}^{\frac{1}{1-\lambda}}[f||h]}{K_{\beta}^{\frac{1}{1-\beta}}[f||h]}=e^{D_{\beta}[f||h]-D_{\lambda}[f||h]}.
\end{equation}
\end{definition}

\noindent \textbf{Remark.} We deduce from Lemma \ref{lem:Renyi} that
	$$
	\frac{N_{\lambda}[f_1^{[h]}]}{N_{\beta}[f_1^{[h]}]}=C^{(D)}_{\lambda,\beta}[f||h],
	$$
	or equivalently
	$$
	C^{(D)}_{\lambda,\beta}[\mathfrak{R}^{-1,[h]}_{1}[f]||h]=\frac{N_{\lambda}[f]}{N_{\beta}[f]},
	$$
which motivates the name we gave to the complexity measure defined in \eqref{eq:complex}. Observe also that the final expression in Eq. \eqref{eq:complex} reminds thus of the LMC-R\'enyi complexity measure, see for example \cite[Section 4]{Puertas2019}.

A first significant property with respect to the divergence transformations is the following equality, which is derived from Eq.~\eqref{eqlemma:appr_diver} after simple algebraic simplifications:
\[
	C^{(D)}_{\lambda,\beta}\left[f^{(\mathfrak A,h)}_{\alpha}||h\right]=\left(C^{(D)}_{\lambda_\alpha,\beta_\alpha}\left[f||h\right]\right)^\alpha.
\]

We next establish the property of monotonicity of the relative complexity measure introduced in Definition \ref{def:complex} with respect to the divergence transformation. This property actually stems from the fact that the divergence transformation is a conjugation of the relative differential-escort and a standard differential-escort transformations, according to Proposition \ref{prop:relative_escort}.
\begin{theorem}\label{th:LMC}
In the same conditions as in Definition \ref{def:complex}, if $f$ and $h$ also satisfy the condition \eqref{cond:KL}, then for any $\gamma$, $\overline{\gamma}\in[0,\infty)$ with $\gamma<\overline{\gamma}$, we have
	\begin{equation}\label{eq:monot}
	C^{(D)}_{\lambda,\beta}\left[\mathfrak{A}^{(h)}_{\gamma}[f]||h\right]\leqslant C^{(D)}_{\lambda,\beta}\left[\mathfrak{A}^{(h)}_{\overline{\gamma}}[f]||h\right].
	\end{equation}
	In particular, for any $\gamma\in(0,1)$ we have
	\begin{equation}\label{eq:monot2}
	C^{(D)}_{\lambda,\beta}\left[\mathfrak{A}^{(h)}_{\gamma}[f]||h\right]\leq C^{(D)}_{\lambda,\beta}[f||h],
	\end{equation}
	while the inequality in \eqref{eq:monot2} is reversed for $\gamma\in(1,\infty)$. The equality is achieved if and only if $f=h$.
\end{theorem}
\begin{proof}
The proof is rather straightforward if we recall the monotonicity property of the LMC-R\'enyi complexity measure with respect to the differential-escort transformation in \cite[Section 4]{Puertas2019}. Indeed, choosing $\gamma>1$ and recalling \eqref{eq:appr_comp} and \eqref{eq:comp_escort}, we obtain
	\begin{equation}\label{eqproof101}
	\begin{split}
	C^{(D)}_{\lambda,\beta}\left[\mathfrak{R}^{-1,[h]}_{1}\mathfrak{E}_{\gamma}\mathfrak{R}^{[h]}_{1}[f]||h\right]&
	=\frac{N_{\lambda}\left[\mathfrak{E}_{\gamma}[f_1^{[h]}]\right]}{N_{\beta}\left[\mathfrak{E}_{\gamma}[f_1^{[h]}]\right]}
	\geq\frac{N_{\lambda}\left[f_1^{[h]}\right]}{N_{\beta}\left[f_1^{[h]}\right]}\\
	&=\frac{K_{\lambda}^{\frac{1}{1-\lambda}}[f||h]}{K_{\beta}^{\frac{1}{1-\beta}}[f||h]}=C^{(D)}_{\lambda,\beta}[f||h],
	\end{split}
	\end{equation}
where the inequality step follows from the choice $\gamma>1$. The proof is similar for $\gamma\in(0,1)$, while the general inequality \eqref{eq:monot} follows from the group structure (provided $\gamma\neq0$) of the divergence transformations with respect to their parameter established in Proposition \ref{prop:group}.
The equality in Eq~\eqref{eqproof101} holds when $f_1^{[h]}=\mathcal U$,  according to~\cite{Puertas2019}. The latter only happens when $f=h$, as trivially follows from Definition~\ref{def:transf}.
\end{proof}
Let us remark that, in the limiting cases $\gamma=1$ and $\gamma=0$, the conjugated transformation reduces to
$$
\mathfrak{R}^{-1,[h]}_{1}\mathfrak{E}_{1}\mathfrak{R}^{[h]}_{1}=\mathcal{I}, \quad \mathfrak{R}^{-1,[h]}_{1}\mathfrak{E}_{0}\mathfrak{R}^{[h]}_{1}[f]=h, \quad {\rm for \ any \ density} \ f,
$$
where $\mathcal{I}$ designs the identity operator; the second equality follows from the fact that $\mathfrak{E}_0$ gives as result a uniform (constant) density on its support, and $\mathfrak{R}^{-1,[h]}_{1}[1]=h$.

\medskip

\noindent \textbf{An application to the Gaussian density.} Let us consider $h(x)=G(x)=\frac{e^{-x^2}}{\sqrt\pi},\; x\in\Rset$. We obtain as a direct consequence of the order relation of the Rényi divergence measures that, for any $\xi\in(0,1)$,
	$$
	\frac{e^{D[f||G]}}{e^{D_\xi[f||G]}}\geqslant 1.
	$$
	On the one hand, we have that
	$$
	D[f||G]=\int_\Rset f(x) \log\left(f(x)\,\sqrt \pi e^{x^2}\right)\,dx=-S[f]+\langle x^2\rangle+\log\sqrt\pi.
	$$
	On the other hand
	$$
	e^{D_\xi[f||G]}=\left(\int_{\Rset} f^\xi(x)G^{1-\xi}(x)dx \right)^{\frac1{\xi-1}}=\sqrt\pi\left(\int_{\Rset} f^\xi(x)e^{-(1-\xi)x^2}dx \right)^{\frac1{\xi-1}}.$$
	We thus deduce from the previous equalities that
	\begin{equation}\label{eq:Gexample}
	\frac{e^{D[f||G]}}{e^{D_\xi[f||G]}}=\frac{e^{-S[f]}e^{\mu_2[f]}}{N_{\xi}^{[G]}[f]}\geqslant 1,
	\end{equation}
	where
	\begin{equation}\label{eq:Gdiv}
	N_\xi^{[G]}[f]:=\left(\int_\Rset f^\xi (x)e^{-(1-\xi) x^2}dx\right)^{\frac{1}{\xi-1}}.
	\end{equation}

As an example of an application of Theorems~\ref{th:divergence} and~\ref{th:LMC}, we remark that the measures defined in Eqs.~\eqref{eq:Gdiv} and ~\eqref{eq:Gexample} are, respectively, monotone with respect to the transformation $\mathfrak A^{(G)}_\gamma$ defined as
$$
\mathfrak A^{(G)}_\gamma[f](y)=C f(x(y))^\gamma e^{\gamma x(y)^2-y^2}, \quad C=\pi^{\frac{\gamma-1}2}\,K_\gamma[G||f],
$$
where the change of variable $y=y(x)$ is given by
$$
\text{erf}(y)=\frac2{C\sqrt{\pi}}\int_{-\infty}^x [f(t)]^{1-\gamma} e^{-\gamma t^2}\,dt-1.
$$ 

\medskip

\noindent \textbf{A relative Fisher measure of complexity.} The next measure of complexity we propose in this paper is constructed having as starting point the relative Fisher divergence introduced in \cite[Section 4.1]{IPT2025} and recalled in Definition \ref{def:relativeFD}. Before introducing its definition, we adopt for simplicity the following notation for some composed transformations that will be employed throughout this section. For any $\alpha$ and $\gamma\in\Rset$, we set
\begin{equation}\label{eq:not_comp}
\mathfrak{D}_{\gamma,\alpha}^{\mathfrak{R},[h]}:=\mathfrak{R}_{\alpha}^{-1,[h]}\circ\mathfrak{D}_{\gamma}\circ\mathfrak{R}_{\alpha}^{[h]}, \quad
\mathfrak{U}_{\gamma,\alpha}^{\mathfrak{R},[h]}:=\mathfrak{R}_{\alpha}^{-1,[h]}\circ\mathfrak{U}_{\gamma}\circ\mathfrak{R}_{\alpha}^{[h]},
\end{equation}
where $\mathfrak{D}_{\gamma}$ and $\mathfrak{U}_{\gamma}$ are the down and up transformations given in Definitions \ref{def:down} and \ref{def:up} respectively, noticing that the two composed transformations in \eqref{eq:not_comp} are mutually inverse, as a simple consequence of Proposition \ref{prop:inv}. We thus define, for $\lambda>\beta$, $\theta\in\Rset\setminus\{0\}$ and $f$, $h$ two differentiable probability density functions satisfying \eqref{cond:support} and $f\neq h$ in at least a set of positive measure,  the following measure of complexity
\begin{equation}\label{eq:comp_Fisher}
C_{\lambda,\beta,\theta}^{(\phi_{\rm rel})}[f||h]:=\frac{\phi_{\lambda,\theta}[f||h]}{\phi_{\beta,\theta}[f||h]},
\end{equation}
where $\phi_{a,b}[f||h]$ denotes the relative Fisher divergence defined in \eqref{eq:relativeFD}. In the previous notation, we prove the following monotonicity result.
\begin{theorem}\label{th:monot_Fisher}
Let $\lambda>\beta$, $\alpha\in(0,\infty)$ and $\gamma\in\Rset\setminus\{2\}$ and let $f$ and $h$ be two probability density functions satisfying \eqref{cond:support}, being differentiable on their support and with the following additional properties:
\begin{itemize}
  \item $\frac{f}{h}$ is a decreasing function on their common support $\Omega$.
  \item $\phi_{\frac{1}{\alpha},\alpha(2-\gamma)}[f||h]<\infty$.
  \end{itemize}
Then, for any $\delta\in(0,1)$, the following inequality holds true:
\begin{equation}\label{eq:monot_Fisher}
C_{\lambda,\beta,\alpha(2-\gamma)}^{(\phi_{\rm rel})}[\mathfrak{U}_{\gamma,\alpha}^{\mathfrak{R},[h]}\circ\mathfrak{A}_{\delta}^{(h)}\circ\mathfrak{D}_{\gamma,\alpha}^{\mathfrak{R},[h]}[f]||h]^{2-\gamma}
\leq C_{\lambda,\beta,\alpha(2-\gamma)}^{(\phi_{\rm rel})}[f||h]^{2-\gamma}.
\end{equation}
The inequality sign is reversed in \eqref{eq:monot_Fisher} if $\delta>1$.
\end{theorem}
\begin{proof}
The general idea of the proof is to connect the relative LMC-R\'enyi complexity measure introduced in Definition \eqref{def:complex} with the relative Fisher complexity measure. We recall from \eqref{eq:Renyi} that
$$
N_{\lambda}[f_{\alpha}^{[h]}]=K_{\xi(\lambda,\alpha)}^{\frac{1}{1-\lambda}}[f||h],
$$
which is equivalent (by taking the inverse of the relative differential-escort transformation) to
\begin{equation}\label{eq:interm10}
N_{\lambda}[f]=K_{\xi(\lambda,\alpha)}^{\frac{1}{1-\lambda}}[\mathfrak{R}_{\alpha}^{-1,[h]}[f]||h],
\end{equation}
where we recall that $\xi(\lambda,\alpha)$ is defined in \eqref{eq:xi(alpha)}. By combining then the identities \eqref{eq:interm10}, \eqref{eq:EF} and \eqref{eq:relFI}, we deduce that
\begin{equation*}
\begin{split}
K_{\xi(\lambda,\alpha)}^{\frac{1}{1-\lambda}}[\mathfrak{D}_{\gamma,\alpha}^{\mathfrak{R},[h]}[f]||h]&=
K_{\xi(\lambda,\alpha)}^{\frac{1}{1-\lambda}}[\mathfrak{R}_{\alpha}^{-1,[h]}\circ\mathfrak{D}_{\gamma}\circ\mathfrak{R}_{\alpha}^{[h]}[f]||h]\\
&=N_{\lambda}[\mathfrak{D}_{\gamma}\circ\mathfrak{R}_{\alpha}^{[h]}[f]]=\phi_{1-\lambda,2-\gamma}^{2-\gamma}[\mathfrak{R}_{\alpha}^{[h]}[f]]\\
&=\left[|\alpha|^{\frac{1}{2-\gamma}}\phi_{1-\lambda,\alpha(2-\gamma)}^{\alpha}[f||h]\right]^{2-\gamma}=|\alpha|\phi_{1-\lambda,\alpha(2-\gamma)}^{\alpha(2-\gamma)}[f||h].
\end{split}
\end{equation*}
We can thus go one step further and calculate the LMC-R\'enyi complexity measure applied to the composed transformation $\mathfrak{D}_{\gamma,\alpha}^{\mathfrak{R},[h]}[f]$. Indeed, we have
\begin{equation*}
\begin{split}
C_{\xi(\lambda,\alpha),\xi(\beta,\alpha)}^{(D)}[\mathfrak{D}_{\gamma,\alpha}^{\mathfrak{R},[h]}[f]||h]
&=\frac{K_{\xi(\lambda,\alpha)}^{\frac{1}{\alpha(1-\lambda)}}[\mathfrak{D}_{\gamma,\alpha}^{\mathfrak{R},[h]}[f]||h]}{K_{\xi(\beta,\alpha)}^{\frac{1}{\alpha(1-\beta)}}[\mathfrak{D}_{\gamma,\alpha}^{\mathfrak{R},[h]}[f]||h]}\\
&=\left[\frac{|\alpha|\phi_{1-\lambda,\alpha(2-\gamma)}^{\alpha(2-\gamma)}[f||h]}{|\alpha|\phi_{1-\beta,\alpha(2-\gamma)}^{\alpha(2-\gamma)}[f||h]}\right]^{\frac{1}{\alpha}}\\
&=C_{1-\lambda,1-\beta,\alpha(2-\gamma)}^{(\phi_{\rm rel})}[f||h]^{2-\gamma}.
\end{split}
\end{equation*}
The latter equality can be also written as
\begin{equation}\label{eq:interm11}
C_{\xi(\lambda,\alpha),\xi(\beta,\alpha)}^{(D)}[f||h]=\left(C_{1-\lambda,1-\beta,\alpha(2-\gamma)}^{(\phi_{\rm rel})}\left[\mathfrak{U}_{\gamma,\alpha}^{\mathfrak{R},[h]}[f]||h\right]\right)^{2-\gamma}.
\end{equation}
Letting next $\delta\in(0,1)$, we infer from the monotonicity property of the relative LMC-R\'enyi complexity measure \eqref{eq:monot2}, \eqref{eq:interm11} and the fact that the composed transformations defined in \eqref{eq:not_comp} are mutually inverse that
\begin{equation*}
		\begin{split}
			\left(C_{1-\lambda,1-\beta,\alpha(2-\gamma)}^{(\phi_{\rm rel})}\left[\mathfrak{U}_{\gamma,\alpha}^{\mathfrak{R},[h]}\circ\mathfrak{A}_{\delta}^{[h]}\circ \mathfrak{D}_{\gamma,\alpha}^{\mathfrak{R},[h]}[f]||h\right]\right)^{2-\gamma}
			&=C_{\xi(\lambda,\alpha),\xi(\beta,\alpha)}^{(D)}\left[\mathfrak{A}_{\delta}^{[h]}\circ \mathfrak{D}_{\gamma,\alpha}^{\mathfrak{R},[h]}[f]||h\right]\\
			&\leq C_{\xi(\lambda,\alpha),\xi(\beta,\alpha)}^{(D)}\left[ \mathfrak{D}_{\gamma,\alpha}^{\mathfrak{R},[h]}[f]||h\right]\\
			&=\left(C_{1-\lambda,1-\beta,\alpha(2-\gamma)}^{(\phi_{\rm rel})}[f||h]\right)^{2-\gamma}.
		\end{split}
\end{equation*}
We thus notice that the inequality \eqref{eq:monot_Fisher} follows from the previous inequality by letting $\lambda':=1-\lambda$, $\beta':=1-\beta$ (and droppìng the primes at the end). The proof is completed by observing that all the transformations employed in the previous steps are well defined. Indeed, one has
\begin{equation*}
\begin{split}
\frac{df_{\alpha}^{[h]}}{dy}&=\frac{\alpha f^{2\alpha-1}(x)}{h^{2\alpha}(x)}\left(\frac{f'(x)}{f(x)}-\frac{h'(x)}{h(x)}\right)\\
&=\frac{\alpha f^{2\alpha-1}(x)}{h^{2\alpha}(x)}\frac{d}{dx}\log\left(\frac{f(x)}{h(x)}\right)<0,
\end{split}
\end{equation*}
since $f/h$ is assumed to be decreasing in $\Omega$ and $\alpha>0$. Thus, the divergence-transformed density $f_{\alpha}^{[h]}$ is decreasing and the down transformation of it is well defined. Moreover, we deduce from \eqref{eq:EF} and \eqref{eq:relFI} that
\begin{equation*}
\begin{split}
N_{\frac{\alpha-1}{\alpha}}[\mathfrak{D}_{\gamma}\mathfrak{R}_{\alpha}^{[h]}[f]]
&=\phi_{\frac{1}{\alpha},2-\gamma}^{2-\gamma}[\mathfrak{R}_{\alpha}^{[h]}[f]]\\
&=\alpha\phi_{\frac{1}{\alpha},(2-\gamma)\alpha}^{(2-\gamma)\alpha}[f||h]<\infty,
\end{split}
\end{equation*}
as assumed in the statement of the theorem. The latter calculation allows thus to apply the inverse relative differential-escort transformation $\mathfrak{R}_{\alpha}^{-1,[h]}$ to the density $\mathfrak{D}_{\gamma}\mathfrak{R}_{\alpha}^{[h]}[f]$, a fact that has been used throughout the proof. The proof is now complete.
\end{proof}
\noindent
\textbf{The minimizers.} Taking into account Theorem~\ref{th:LMC}, the equality holds when  $ \mathfrak{D}_{\gamma,\alpha}^{\mathfrak{R},[h]}[f]=h,$ or equivalently, when $f= \mathfrak{U}_{\gamma,\alpha}^{\mathfrak{R},[h]}[h].$ More precisely
$$
f_{\rm min}=\mathfrak{R}_{\alpha}^{-1,[h]}\circ\mathfrak{U}_{\gamma}\circ\mathfrak{R}_{\alpha}^{[h]}[h]=\mathfrak{R}_{\alpha}^{-1,[h]}\circ\mathfrak{U}_{\gamma}\left[\mathcal U\right ]=\mathfrak{R}_{\alpha}^{-1,[h]}[\rho_\gamma],
$$
where $\mathcal U$ denotes the uniform density $\mathcal U(x)=1,\,x\in(0,1)$ and where, for $\gamma\in\mathbb R\setminus\{1,2\}$,  $\rho_\gamma$ is given by
$$
\rho_\gamma(u)=|\gamma-2|^\frac{1}{2-\gamma} \left(1+\frac{\sign(\gamma-2)(1-\gamma)}{|\gamma-2|^{\frac{\gamma-1}{\gamma-2}}}\,u\right)_+^{\frac1{1-\gamma}},
$$
for
$$
u\in\begin{cases}
	\left[0,|\gamma-2|^{\frac1{\gamma-2}}\frac{\gamma-2}{\gamma-1}\right],\quad &{\rm if}\, \gamma\in(-\infty,1)\cup(2,\infty),\\[2mm]
	[0,\infty),\quad&{\rm if}\,\gamma\in(1,2).
\end{cases}
$$
For $\gamma=2$ we obtain
$$
\rho_2(u)=\left(\frac{1}{e-u}\right)_+,\quad u\in[0,e-1]
$$
and for $\gamma=1$
$$
\rho_1(u)=e^{-u},\quad u\in[0,\infty).
$$
Note that the support of $\rho_\gamma$ is infinite for any $\gamma\in[1,2)$ and finite in the opposite case.

\medskip

\noindent \textbf{A complexity measure related to the relative cumulative moments.} We introduce next another measure of complexity, built on the relative cumulative moments introduced in \cite{IPT2025} and recalled in Definition \ref{def:relativeCM}. We thus define, for $\lambda,\beta\in\mathbb R\setminus\{0\}$, $\gamma\in\Rset\setminus\{2\}$ and $\alpha\in\Rset$ such that
\begin{equation}\label{cond:compl_mom}
\sign(\lambda-\beta)=\sign(\alpha)\sign(\gamma–2)
\end{equation}
and $f$, $h$ probability density functions satisfying \eqref{cond:support}, the following measure:
\begin{equation}\label{eq:comp_moment}
C_{\lambda,\beta,\gamma,\alpha}^{(\sigma_{\rm rel})}[f||h]:=\frac{\mu_{\lambda,\alpha}^{\frac{1}{(\gamma-2)\lambda}}[f||h]}{\mu_{\beta,\alpha}^{\frac{1}{(\gamma-2)\beta}}[f||h]}.
\end{equation}
The next result states the monotonicity property of the complexity measure introduced in \eqref{eq:comp_moment}.
\begin{theorem}\label{th:monot_moment}
Let $\lambda,\beta\in\mathbb R\setminus\{0\}$, $\gamma\in\Rset\setminus\{2\}$, $\alpha\in\Rset\setminus\{0\}$ such that the condition~\eqref{cond:compl_mom} is fulfilled, and let $f$, $h$ be two probability density functions satisfying \eqref{cond:support}. Then, for any $\delta\in(0,1)$, the following inequality holds true:
\begin{equation}\label{eq:monot_moment}
C_{\lambda,\beta,\gamma,\alpha}^{(\sigma_{\rm rel})}[\mathfrak{D}_{\gamma,\alpha}^{\mathfrak{R},[h]}\circ\mathfrak{A}_{\delta}^{(h)}\circ\mathfrak{U}_{\gamma,\alpha}^{\mathfrak{R},[h]}[f]||h]^{\frac{1}{\alpha}}
\leq C_{\lambda,\beta,\gamma,\alpha}^{(\sigma_{\rm rel})}[f||h]^{\frac{1}{\alpha}}.
\end{equation}
\end{theorem}
\begin{proof}
We follow a similar scheme as the one employed in the proof of Theorem \ref{th:monot_Fisher}. We infer from \eqref{eq:interm10} and \eqref{eq:ME} that
\begin{equation*}
\begin{split}
K_{\xi(\lambda,\alpha)}^{\frac{1}{1-\lambda}}[\mathfrak{R}_{\alpha}^{-1,[h]}\circ\mathfrak{U}_{\gamma}\circ\mathfrak{R}_{\alpha}^{[h]}[f]||h]
&=N_{\lambda}[\mathfrak{U}_{\gamma}\circ\mathfrak{R}_{\alpha}^{[h]}[f]]\\
&=\left(|2-\gamma|\sigma_{\frac{\lambda-1}{2-\gamma}}[\mathfrak{R}_{\alpha}^{[h]}[f]]\right)^{\frac{1}{\gamma-2}}\\
&=|2-\gamma|^{\frac{1}{\gamma-2}}\mu_{\frac{\lambda-1}{2-\gamma},\alpha}^{\frac{1}{1-\lambda}}[f||h].
\end{split}
\end{equation*}
Using this identity and \eqref{eq:complex}, let us pick $\delta\in(0,1)$ and continue by evaluating, on the one hand,
\begin{equation}\label{eq:interm13}
\begin{split}
C_{\xi(\lambda,\alpha),\xi(\beta,\alpha)}^{(D)}[\mathfrak{A}_{\delta}^{(h)}[f]||h]
&=C_{\xi(\lambda,\alpha),\xi(\beta,\alpha)}^{(D)}[\mathfrak{U}_{\gamma,\alpha}^{\mathfrak{R},[h]}\circ\mathfrak{D}_{\gamma,\alpha}^{\mathfrak{R},[h]}\circ\mathfrak{A}_{\delta}^{(h)}[f]||h]\\
&=\frac{K_{\xi(\lambda,\alpha)}^{\frac{1}{1-\xi(\lambda,\alpha)}}[\mathfrak{U}_{\gamma,\alpha}^{\mathfrak{R},[h]}\circ\mathfrak{D}_{\gamma,\alpha}^{\mathfrak{R},[h]}\circ\mathfrak{A}_{\delta}^{(h)}[f]||h]}
{K_{\xi(\beta,\alpha)}^{\frac{1}{1-\xi(\beta,\alpha)}}[\mathfrak{U}_{\gamma,\alpha}^{\mathfrak{R},[h]}\circ\mathfrak{D}_{\gamma,\alpha}^{\mathfrak{R},[h]}\circ\mathfrak{A}_{\delta}^{(h)}[f]||h]}\\
&=\left(\frac{\mu_{\frac{\lambda-1}{2-\gamma},\alpha}^{\frac{1}{1-\lambda}}
[\mathfrak{D}_{\gamma,\alpha}^{\mathfrak{R},[h]}\circ\mathfrak{A}_{\delta}^{(h)}[f]||h]}
{\mu_{\frac{\beta-1}{2-\gamma},\alpha}^{\frac{1}{1-\beta}}[\mathfrak{D}_{\gamma,\alpha}^{\mathfrak{R},[h]}\circ\mathfrak{A}_{\delta}^{(h)}[f]||h]}\right)^{\frac{1}{\alpha}}\\
&=C_{\frac{\lambda-1}{2-\gamma},\frac{\beta-1}{2-\gamma},\gamma,\alpha}^{(\sigma_{\rm rel})}[\mathfrak{D}_{\gamma,\alpha}^{\mathfrak{R},[h]}\circ\mathfrak{A}_{\delta}^{(h)}[f]||h]^{\frac{1}{\alpha}}.
\end{split}
\end{equation}
On the other hand, we have after a similar calculation (but easier than the one before) stemming again from \eqref{eq:complex} that
\begin{equation*}
C_{\xi(\lambda,\alpha),\xi(\beta,\alpha)}^{(D)}[\mathfrak{U}_{\gamma,\alpha}^{\mathfrak{R},[h]}[f]||h]=
C_{\frac{\lambda-1}{2-\gamma},\frac{\beta-1}{2-\gamma},\gamma,\alpha}^{(\sigma_{\rm rel})}[f||h]^{\frac{1}{\alpha}},
\end{equation*}
which implies
\begin{equation}\label{eq:interm12}
C_{\xi(\lambda,\alpha),\xi(\beta,\alpha)}^{(D)}[f||h]=C_{\frac{\lambda-1}{2-\gamma},\frac{\beta-1}{2-\gamma},\gamma,\alpha}^{(\sigma_{\rm rel})}[\mathfrak{D}_{\gamma,\alpha}^{\mathfrak{R},[h]}[f]||h]^{\frac{1}{\alpha}}.
\end{equation}
Since $\delta\in(0,1)$, we infer from the monotonicity inequality \eqref{eq:monot2} satisfied by the relative LMC-R\'enyi complexity measure, \eqref{eq:interm12} and \eqref{eq:interm13} that
\begin{equation}\label{eq:interm14}
C_{\frac{\lambda-1}{2-\gamma},\frac{\beta-1}{2-\gamma},\gamma,\alpha}^{(\sigma_{\rm rel})}[\mathfrak{D}_{\gamma,\alpha}^{\mathfrak{R},[h]}\circ\mathfrak{A}_{\delta}^{(h)}[f]||h]^{\frac{1}{\alpha}}\leq C_{\frac{\lambda-1}{2-\gamma},\frac{\beta-1}{2-\gamma},\gamma,\alpha}^{(\sigma_{\rm rel})}[\mathfrak{D}_{\gamma,\alpha}^{\mathfrak{R},[h]}[f]||h]^{\frac{1}{\alpha}}.
\end{equation}
Applying \eqref{eq:interm14} to a transformed density $\mathfrak{U}_{\gamma,\alpha}^{\mathfrak{R},[h]}[f]$ instead of $f$ and recalling that the composed transformations introduced in \eqref{eq:not_comp} are mutually inverse, we obtain that
\begin{equation}\label{eq:interm15}
C_{\frac{\lambda-1}{2-\gamma},\frac{\beta-1}{2-\gamma},\gamma,\alpha}^{(\sigma_{\rm rel})}[\mathfrak{D}_{\gamma,\alpha}^{\mathfrak{R},[h]}\circ\mathfrak{A}_{\delta}^{(h)}\circ\mathfrak{U}_{\gamma,\alpha}^{\mathfrak{R},[h]}[f]||h]^{\frac{1}{\alpha}}
\leq C_{\frac{\lambda-1}{2-\gamma},\frac{\beta-1}{2-\gamma},\gamma,\alpha}^{(\sigma_{\rm rel})}[f||h]^{\frac{1}{\alpha}}.
\end{equation}
The proof is now completed by adopting the notation
$$
\lambda':=\frac{\lambda-1}{2-\gamma}, \quad \beta':=\frac{\beta-1}{2-\gamma}
$$
and observing that, in this notation,
$$
\frac{1}{1-\lambda}=\frac{1}{\lambda'(2-\gamma)}, \quad \frac{1}{1-\beta}=\frac{1}{\beta'(2-\gamma)},
$$
and finally dropping the primes from the notation.
\end{proof}
\noindent
\textbf{The minimizers.} The inequality~\eqref{eq:monot_moment} is sharp but has no exact minimizers. This is because a minimizer, if existing, must satisfy
$f_{\rm min}= \mathfrak{D}_{\gamma,\alpha}^{\mathfrak{R},[h]}[h],$ or equivalently
$$
f_{\rm min}=\mathfrak{R}_{\alpha}^{-1,[h]}\circ\mathfrak{D}_{\gamma}\circ\mathfrak{R}_{\alpha}^{[h]}[h]=\mathfrak{R}_{\alpha}^{-1,[h]}\circ\mathfrak{D}_{\gamma}[\mathcal U],
$$
recalling that $\mathcal U$ is the uniform density. But the latter minimizer is not well-defined, since
the down transform only applies to decreasing functions. As analyzed in~\cite{IP2025}, the down transformation of a sequence of decreasing densities which tend to a uniform one, approximates a Dirac delta distribution. This approximation proves the sharpness of the inequality~\eqref{eq:monot_moment}.

\medskip

\noindent \textbf{Remark.} It is obvious that the composed transformations
$$
\mathfrak{D}_{\gamma,\alpha}^{\mathfrak{R},[h]}\circ\mathfrak{A}_{\delta}^{(h)}\circ\mathfrak{U}_{\gamma,\alpha}^{\mathfrak{R},[h]} \quad {\rm and} \quad \mathfrak{U}_{\gamma,\alpha}^{\mathfrak{R},[h]}\circ\mathfrak{A}_{\delta}^{(h)}\circ\mathfrak{D}_{\gamma,\alpha}^{\mathfrak{R},[h]}
$$
appearing in Theorems \ref{th:monot_Fisher} and \ref{th:monot_moment} have a group structure with respect to $\delta\neq0$, inherited from the similar group structure of the divergence transformation $\mathfrak{A}_{\delta}^{(h)}$ established in Proposition \ref{prop:group}.

\section{Divergence transformations of usual densities}\label{sec:examples}

We gather in this section the explicit calculation of the transformed densities of several usual probability density functions employed frequently in Information Theory, as examples of applications of the divergence transformation.

\subsection{The densities $\pdf(x)=(n+1)x^n$ and $\pdfr(x)=1$}

Let us first consider, as a toy example, the following pair of probability densities: $\pdf(x)=(n+1)x^n$ with $n>0$ and $\pdfr(x)=1$ with $x\in(0,1)$. We first compute
\begin{equation*}
	H(y)=\int_0^y\pdfr(s)ds=\int_0^y 1\,ds=y
\end{equation*}
In the next step, we calculate the function $K_\alpha[\pdfr||\pdf]$:
\begin{equation*}
	K_\alpha[\pdfr||\pdf]=\int_0^1(n+1)^{1-\alpha} t^{(1-\alpha)n}\cdot 1^\alpha dt =\frac{(n+1)^{1-\alpha}}{1+(1-\alpha)n}
\end{equation*}
provided that $(1-\alpha)n>-1$, which is equivalent to $\alpha<1+\frac{1}{n}$. We finally compute the incomplete integral $K(\alpha;x)$ as follows:
\begin{equation*}
	K(1-\alpha;x)=\int_0^x(n+1)^{1-\alpha} t^{(1-\alpha)n}\cdot 1^\alpha dt =\frac{(n+1)^{1-\alpha}}{1+(1-\alpha)n}\,x^{1+(1-\alpha)n}.
\end{equation*}
The change of variable in \eqref{eq:y(x)0} is given by
$$
y=x^{1+(1-\alpha)n}, \quad {\rm or, \ equivalently,} \quad x(y)=y^{\frac1{1+(1-\alpha)n}}.
$$
We insert the previous calculations into \eqref{eq:tilderho} to find
\begin{equation*}
	\pdf_{\alpha}^{(\pdfr)}(y)=K_\alpha[\pdfr||\pdf]\left(\frac{\pdf(x(y))}{1}\right)^\alpha=\frac{(n+1)^{1-\alpha}}{1+(1-\alpha)n}\left((n+1)x(y)^n\right)^\alpha
\end{equation*}
which gives, after simplifications,
\begin{equation*}
	\pdf_{\alpha}^{(\pdfr)}(y)=\frac{n+1}{1+(1-\alpha)n}\, y^{\frac{n\alpha}{1+(1-\alpha)n}}
\end{equation*}
Note that, in this simple calculation, the critical values defined in \eqref{eq:asup} are given by $\alpha_{c}^{-}=-\infty$ and $\alpha_c^{+}=1+\frac{1}{n}$.

\subsection{Exponential densities}

Let us next consider the following exponential densities: $\pdfr(x)=e^{-x}$ and $\pdf(x)=\gamma e^{-\gamma x}$ with $\gamma\in(0,\infty)\setminus\{1\}$ and $x\in(0,\infty)$. We proceed as in the previous section, starting from the calculation of
\begin{equation*}
	H(y)=\int_0^y e^{-t}dt=1-e^{-y}.
\end{equation*}
In a second step, we calculate the function $K_\alpha[\pdfr||\pdf]$, provided that $\alpha-\gamma(\alpha-1)>0:$
\begin{equation*}
K_\alpha[\pdfr||\pdf]=\int_0^\infty \gamma^{1-\alpha} e^{\gamma (\alpha-1)t} e^{-\alpha t}dt =\gamma^{1-\alpha} \int_0^\infty e^{-[(1-\alpha)\gamma+\alpha]t} \,dt =\frac{\gamma^{1-\alpha}}{(1-\alpha)\gamma+\alpha}
\end{equation*}
We next compute the incomplete integral $K(1-\alpha;x)$, which gives
\begin{eqnarray*}
	K(1-\alpha;x)&=&	\int_0^x \pdf(t)^{1-\alpha}\pdfr(t)^\alpha dt=\gamma^{1-\alpha}\int_0^x e^{-[(1-\alpha)\gamma+\alpha] t}\,dt
	\\
	&=&\frac{\gamma^{1-\alpha}}{(1-\alpha)\gamma+\alpha}\left(1-e^{-[(1-\alpha)\gamma+\alpha]x}\right).
\end{eqnarray*}
We next obtain from \eqref{eq:y(x)0} that
\begin{equation*}
	1-e^{-y}=1-e^{-[(1-\alpha)\gamma+\alpha]x}, \quad {\rm or, \ equivalently}, \quad x(y)=\frac{y}{(1-\alpha)\gamma+\alpha}.
\end{equation*}
We finally introduce in \eqref{eq:tilderho} the previously calculated functions to obtain
\begin{equation*}
 \Appr{\pdfr}{\alpha}(y)=K_{\alpha}[h||f]\left(\frac{\pdf(x(y))}{\pdfr(x(y))}\right)^\alpha\,\pdfr(y)=\frac{\gamma}{(1-\alpha)\gamma+\alpha} e^{-(\gamma-1)\alpha x(y)}e^{-y},
\end{equation*}
which gives, after simplifications,
\begin{equation*}
	\Appr{\pdfr}{\alpha}(y)=\frac{\gamma}{(1-\alpha)\gamma+\alpha} e^{-\frac{\gamma}{(1-\alpha)\gamma+\alpha}y}.
\end{equation*}
In this case, we observe that the critical values defined in \eqref{eq:asup} are $\alpha_c^{-}=-\infty$ and $\alpha_c^{+}=\frac{\gamma}{\gamma-1}$, if $\gamma>1$, respectively $\alpha_c^{-}=-\frac{\gamma}{1-\gamma}$ and $\alpha_c^{+}=\infty$, if $\gamma\in(0,1)$.

\subsection{Exponential and Gaussian}\label{sec:expGauss}

Let us consider next $\pdfr(x)=\frac2{\sqrt{\pi}}e^{-x^2}$ and $\pdf(x)=\gamma e^{-\gamma x}$, with $\gamma>0$ and defined for $x\in(0,\infty)$. We proceed as before by calculating the functions involved in the definition of the divergence transformation. We start with
\begin{equation*}
	H(y)=\int_0^y \frac2{\sqrt \pi}e^{-t^2}dt= \text{erf}(y).
\end{equation*}
We next calculate the function $K_\alpha[\pdfr||\pdf]$:
\begin{align*}
	K_\alpha[\pdfr||\pdf]&=\int_0^\infty \gamma^{1-\alpha} e^{-\gamma (1-\alpha)t} \left(\frac2{\sqrt \pi}e^{-t^2}\right)^{\alpha}\, dt =\frac{\gamma^{1-\alpha}2^{\alpha}}{\pi^{\frac{\alpha}{2}} } \int_0^\infty e^{-[\gamma(1-\alpha)t+\alpha t^2] } \,dt \\[0.5em]
	& = \frac{\gamma^{1-\alpha}2^{\alpha-1}}{\pi^{\frac{\alpha-1}{2}}\alpha^{\frac{1}{2}} }e^{\frac{\gamma^2}{4\alpha}(1-\alpha)^2}\text{erfc}\left(\frac{\gamma}{2\sqrt{\alpha}}(1-\alpha)\right).
\end{align*}
provided that $\alpha>0$. It remains to compute the incomplete integral $K(1-\alpha;x)$, as follows.
\begin{eqnarray*}
	K(1-\alpha;x)&=&	\int_0^x [\pdf(t)]^{1-\alpha}[\pdfr(t)]^\alpha dt= \int_{0}^{x} \gamma^{1-\alpha} e^{-\gamma (1-\alpha)t} \left(\frac2{\sqrt \pi}e^{-t^2}\right)^{\alpha}\,dt
	\nonumber \\[0.5em]
	&=&  \frac{\gamma^{1-\alpha}2^{\alpha-1}}{\pi^{\frac{\alpha-1}{2}}\alpha^{\frac{1}{2}} }e^{\frac{\gamma^2}{4\alpha}(1-\alpha)^2}\left[\text{erf}\left(\sqrt{\alpha}\, x + \frac{\gamma}{2\sqrt{\alpha} }(1-\alpha)  \right) - \text{erf}\left(\frac{\gamma}{2\sqrt{\alpha} }(1-\alpha) \right) \right].
\end{eqnarray*}
By substituting the previous calculations in~\eqref{eq:y(x)}, we find
\begin{equation*}
\begin{split}
\text{erf}(y)&=
\frac{\text{erf}\left(\sqrt{\alpha}x+\frac{\gamma}{2\sqrt{\alpha}(1-\alpha)}\right)-\text{erf}\left(\frac{\gamma}{2\sqrt{\alpha}(1-\alpha)}\right)}{\text{erfc}\left(\frac{\gamma}{2\sqrt{\alpha}(1-\alpha)}\right)}\\
&= 1 - \frac{\text{erfc}\left(\sqrt{\alpha}\, x + \frac{\gamma}{2\sqrt{\alpha}} (1-\alpha) \right)}{\text{erfc}\left(\frac{\gamma}{2\sqrt{\alpha}} (1-\alpha) \right)},
\end{split}
\end{equation*}
and, taking into account that ${\rm erf}(y)=1-{\rm erfc}(y)$, we further deduce that
\begin{equation*}
\text{erfc}\left(\sqrt{\alpha}\, x + \frac{\gamma}{2\sqrt{\alpha}} (1-\alpha) \right)=\mathcal C(\gamma,\alpha)\text{erfc}(y), \quad \mathcal C(\gamma,\alpha):=\text{erfc}\left(\frac{\gamma}{2\sqrt{\alpha}} (1-\alpha) \right).
\end{equation*}	
The latter equality implies
\begin{equation}\label{eq:interm17}
x(y)=\frac{{\rm erfc}^{-1}\left(\mathcal C(\gamma,\alpha)\, {\rm erfc}(y)\right)}{\sqrt \alpha} + \frac{\gamma}{2\alpha} (\alpha-1).
\end{equation}
We finally calculate the divergence transformation according to \eqref{eq:tilderho}, that is,
\begin{equation*}
	\begin{split}
		\Appr{\pdfr}{\alpha}(y)&=K_\alpha[\pdfr||\pdf]\left(\frac{\pdf(x(y))}{\pdfr(x(y))}\right)^\alpha\,\pdfr(y)\\
		&=\frac{\gamma}{\sqrt{\alpha}}e^{\frac{\gamma^2}{4\alpha}(1-\alpha)^2} \text{erfc}\left(\frac{\gamma}{2\sqrt{\alpha }}(1-\alpha) \right)e^{-\alpha \gamma x(y) + \alpha x^2(y)} e^{-y^2}.
	\end{split}
\end{equation*}
Introducing the notation
$$
\xi(y):={\rm erfc}^{-1}\left(\mathcal C(\gamma,\alpha)\,{\rm erfc}(y)\right),
$$
we infer from \eqref{eq:interm17} by direct calculation that
\begin{equation*}
\begin{split}
-\alpha\gamma x(y)+\alpha x^2(y)&=\xi(y)^2+\frac{\gamma(\alpha-1)\xi(y)}{\sqrt{\alpha}}+\frac{\gamma^2(\alpha-1)^2}{4\alpha}-\frac{\gamma\alpha\xi(y)}{\sqrt{\alpha}}-\frac{\gamma^2(\alpha-1)}{2}\\
&=\xi(y)^2-\frac{\gamma}{\sqrt{\alpha}}\xi(y)+\frac{\gamma^2(1-\alpha^2)}{4\alpha},
\end{split}
\end{equation*}
whence
\begin{equation*}
\Appr{\pdfr}{\alpha}(y)= \widetilde{\mathcal C}\exp\left(\xi(y)^2-\frac{\gamma}{\sqrt{\alpha}}\xi(y)\right)e^{-y^2},
\end{equation*}
with
$$
\widetilde{\mathcal C}:=\frac{\gamma}{\sqrt{\alpha}}\exp\left(\frac{\gamma^2}{2\alpha}(1-\alpha)\right)\text{erfc}\left(\frac{\gamma}{2\sqrt{\alpha }}(1-\alpha) \right).
$$
\noindent
The critical values defined in \eqref{eq:asup} are $\alpha_c^{-}=0$ and $\alpha_c^{+}=\infty$, since it is obvious that the improper integral defining $K_{\alpha}[h||f]$ diverges for any $\alpha<0$. Note that, in this case, there is no concentration to a Dirac distribution as $\alpha\to\alpha_c^{-}=0$, since $K_{0}[h||f]=1$.

\subsection{Two Gaussian densities}\label{subsec:twogauss}

Let us compute next the divergence transformation by taking as reference function the standard Gaussian density function $\pdfr(x)=\frac{1}{\sqrt \pi}e^{-x^2}$ and as function to be transformed a Gaussian with a different variance $\sigma>0$ and mean $\mu$
$$
\pdf(x)=\frac{1}{\sqrt{2\pi \sigma^2}}e^{-\frac{(x-\mu)^2}{2\sigma^2}}, \quad \sigma \neq \frac{1}{\sqrt{2}}\, ,
$$
with $x\in\Rset$. Following the same steps as in the previous sections, we have

\begin{equation*}
	H(y)=\int_{-\infty}^y \frac1{\sqrt \pi}e^{-t^2}\,dt= \frac{\text{erf}(y)+1}2.
	\end{equation*}
	and

	\begin{align*}
		K_\alpha[\pdfr||\pdf]&=\int_{-\infty}^\infty \left(\frac{1}{\sqrt{2\pi \sigma^2}}e^{-\frac{(t-\mu)^2}{2\sigma^2}} \right)^{1-\alpha} \left(\frac1{\sqrt \pi}e^{-t^2}\right)^{\alpha}\, dt \\
		& = \frac{2^{\frac{\alpha-1}{2}}}{\pi^{\frac{1}{2}}\sigma^{1-\alpha}  } e^{-\frac{\alpha(1-\alpha)\mu^2 }{1+\alpha(2\sigma^2-1) }} \int_{-\infty}^\infty e^{-\frac{1+\alpha(2\sigma^2-1) }{2\sigma^2 }\left[t - \frac{\mu(1-\alpha)}{1+\alpha(2\sigma^2-1)} \right]^2 } \,dt\\[0.5em]
		& = \frac{2^{\frac{\alpha}{2}}\sigma^{\alpha} }{\sqrt{1+\alpha(2\sigma^2-1) }  }e^{-\frac{\alpha(1-\alpha)\mu^2 }{1+\alpha(2\sigma^2-1) }},
	\end{align*}
provided that $1+\alpha(2\sigma^2-1)>0$ or, equivalently, $\alpha(1-2\sigma^2)<1$. We continue with the calculation of $K(1-\alpha;x)$:
\begin{eqnarray*}
		K(1-\alpha;x)&=& \int_{-\infty}^{x} \left(\frac{1}{\sqrt{2\pi \sigma^2}}e^{-\frac{(t-\mu)^2}{2\sigma^2}} \right)^{1-\alpha} \left(\frac1{\sqrt \pi}e^{-t^2}\right)^{\alpha}\,dt
		\nonumber \\
			& = &\frac{2^{\frac{\alpha-1}{2}}}{\pi^{\frac{1}{2}}\sigma^{1-\alpha}  } e^{-\frac{\alpha(1-\alpha)\mu^2 }{1+\alpha(2\sigma^2-1) }} \int_{-\infty}^x e^{-\frac{1+\alpha(2\sigma^2-1) }{2\sigma^2 }\left[t - \frac{\mu(1-\alpha)}{1+\alpha(2\sigma^2-1)} \right]^2 } \,dt\\[0.5em]
			\nonumber \\
			& = &\frac{2^{\frac{\alpha-1}{2}}}{\pi^{\frac{1}{2}}\sigma^{1-\alpha}  } e^{-\frac{\alpha(1-\alpha)\mu^2 }{1+\alpha(2\sigma^2-1) }} \frac{\sqrt{\pi} \sqrt{2}\sigma}{2\sqrt{1+\alpha(2\sigma^2-1)}}\\
			&\times&\left\{1+{\rm erf}\left[\frac{\sqrt{1+\alpha(2\sigma^2-1)}}{\sqrt 2\,\sigma}\left(x-\frac{\mu(1-\alpha)}{1+\alpha(2\sigma^2-1)}\right)\right]\right\}
			\nonumber \\
			& = &\frac{2^{\frac{\alpha-2}{2}}\sigma^{\alpha} e^{-\frac{\alpha(1-\alpha)\mu^2 }{1+\alpha(2\sigma^2-1) }}}{\sqrt{1+\alpha(2\sigma^2-1)}}\left(1+{\rm erf}(u(x))\right)
\end{eqnarray*}
where
$$
u(x):=\sqrt{\frac{1+\alpha(2\sigma^2-1)}{2\sigma^2}}\left(x - \frac{\mu(1-\alpha)}{1+\alpha(2\sigma^2-1) } \right).
$$
We can thus substitute the previous calculations in \eqref{eq:y(x)0} to get
\begin{equation*}
		\frac{\text{erf}(y)+1}2 = \frac{1}{2}[\text{erf}(u(x)) +1],
	\end{equation*}
	from where
\begin{equation*}
	y(x)=u(x)
\end{equation*}
or, equivalently,
\begin{equation}\label{eq:gaussxy}
	x(y) = \sqrt{\frac{2\sigma^2}{1+\alpha(2\sigma^2-1)}}\,y + \frac{\mu(1-\alpha)}{1+\alpha(2\sigma^2-1) }.
\end{equation}
In order to conclude, we insert \eqref{eq:gaussxy} into \eqref{eq:tilderho} and obtain
\begin{align*}
	&\Appr{\pdfr}{\alpha}(y)=K_\alpha[\pdfr||\pdf]\left(\frac{\pdf(x(y))}{\pdfr(x(y))}\right)^\alpha\,\pdfr(y) \\
	&= \frac{1}{\sqrt{\pi(1+\alpha(2\sigma^2-1))}}\exp\left\{-\frac{\alpha(1-\alpha)\mu^2}{1+\alpha(2\sigma^2-1)}\right\} \exp\left\{-\frac{\alpha}{2\sigma^2}(x(y)-\mu)^2+\alpha x^2(y)-y^2\right\} \\
& =\mathcal{K}(\mu,\sigma,\alpha)\exp\left\{-\frac{y^2}{1+\alpha(2\sigma^2-1)}+\frac{2\sqrt{2}\alpha\mu\sigma}{(1+\alpha(2\sigma^2-1))^{\frac{3}{2}}}y\right\},
\end{align*}
where
$$
\mathcal{K}(\mu,\sigma,\alpha):=\frac{1}{\sqrt{\pi(1+\alpha(2\sigma^2-1))}}\exp\left\{-\frac{2\mu^2\sigma^2\alpha^2}{(1+\alpha(2\sigma^2-1))^2}\right\}.
$$
We finally observe that, when $\alpha \to \alpha_c^{+}$ or $\alpha\to\alpha_c^{-}$, where
\begin{equation}\label{eq:acrit}
\alpha_c^{+}= \begin{cases}
	\frac{1}{1-2\sigma^2}, &  {\rm if } \,\, 2\sigma^2 < 1  \\
	\infty , &  {\rm if } \,\, 2\sigma^2 > 1 \, ,
\end{cases} \quad
\alpha_c^{-}=\begin{cases}
	-\infty, &  {\rm if } \,\, 2\sigma^2 < 1  \\
	-\frac{1}{2\sigma^2-1}, &  {\rm if } \,\, 2\sigma^2 > 1 \, ,
\end{cases}
\end{equation}
the transformed pdf $\Appr{\pdfr}{\alpha}(y)$ concentrates at $y=0$ in the form of a Dirac mass.

\subsection{Beta distributions}

In this final example we consider two Beta distributions
$$
f(x;\beta_1,\beta_2) = \frac{1}{B(\beta_1,\beta_2)}x^{\beta_1 -1}(1-x)^{\beta_2 -1}, \quad
h(x;\gamma_1,\gamma_2) = \frac{1}{B(\gamma_1,\gamma_2)}x^{\gamma_1 -1}(1-x)^{\gamma_2 -1},
$$
defined for $x\in(0,1)$ and parameters $\beta_1$, $\beta_2$, $\gamma_1$, $\gamma_2>0$, where
$$
B(z_1,z_2) = \int_{0}^{1} t^{z_1-1}(1-t_1)^{z_2-1}\, dt, \quad z_1, \ z_2>0,
$$
is the usual Beta function. We follow the same steps as in the previous sections, but, as we shall see, the calculations are much more involved. We start by computing
\begin{equation*}
	H(y)=\int_{0}^y \frac{1}{B(\gamma_1,\gamma_2)}t^{\gamma_1 -1}(1-t)^{\gamma_2 -1} \, dt = \frac{B(y;\gamma_1,\gamma_2)}{B(\gamma_1,\gamma_2)},
\end{equation*}
where
$$
B(z;\alpha_1,\alpha_2) = \int_{0}^{z} t^{\alpha_1-1}(1-t)^{\alpha_2-1}\,dt
$$
designs the incomplete Beta function. We next compute $K_\alpha[\pdfr||\pdf]$ as follows:
\begin{align*}
	K_\alpha[\pdfr||\pdf]&=\int_{0}^1 \left(  \frac{1}{B(\beta_1,\beta_2)}x^{\beta_1 -1}(1-x)^{\beta_2 -1} \right)^{1-\alpha} \left( \frac{1}{B(\gamma_1,\gamma_2)}x^{\gamma_1 -1}(1-x)^{\gamma_2 -1}  \right)^{\alpha}\, dx \\[0.5em]
	& =  \frac{1}{B(\beta_1,\beta_2)}\left[ \frac{ B(\beta_1,\beta_2)}{ B(\gamma_1,\gamma_2)}\right]^{\alpha} \int_{0}^{1} x^{(1-\alpha)(\beta_1 -1)+\alpha(\gamma_1 -1)}(1-x)^{(1-\alpha)(\beta_2-1)+\alpha(\gamma_2-1)}\, dx \\[0.5em]
	& = \frac{1}{B(\beta_1,\beta_2)}\left[ \frac{ B(\beta_1,\beta_2)}{ B(\gamma_1,\gamma_2)}\right]^{\alpha} \int_{0}^{1} x^{\alpha(\gamma_1 - \beta_1) + \beta_1 - 1} (1-x)^{\alpha(\gamma_2-\beta_2)+\beta_2 -1}\, dx \\[0.5em]
	& = \frac{1}{B(\beta_1,\beta_2)}\left[ \frac{ B(\beta_1,\beta_2)}{ B(\gamma_1,\gamma_2)}\right]^{\alpha}  \int_{0}^{1} x^{\delta_1(\alpha) - 1} (1-x)^{\delta_2(\alpha) -1}\, dx \\[0.5em]
	& = \frac{1}{B(\beta_1,\beta_2)}\left[ \frac{ B(\beta_1,\beta_2)}{ B(\gamma_1,\gamma_2)}\right]^{\alpha} B(\delta_1(\alpha),\delta_2(\alpha)) \, ,
\end{align*}
where we have adopted the notation
$$
\delta_1(\alpha):= \alpha(\gamma_1 - \beta_1) + \beta_1, \quad  \delta_2(\alpha):= \alpha(\gamma_2-\beta_2)+\beta_2.
$$
Note that the functions $\delta_1(\alpha)$ and $\delta_2(\alpha)$ provide a linear interpolation between $\beta_{1,2}$ and $\gamma_{1,2}$ when $\alpha\in[0,1]$. Taking into account that the Beta function is defined for positive parameters, we assume that $\delta_1(\alpha)>0$ and $\delta_2(\alpha)>0$, which gives
\begin{equation}\label{eq:interm16}
\alpha(\gamma_1 - \beta_1) + \beta_1> 0 \quad {\rm and} \quad \alpha(\gamma_2-\beta_2)+\beta_2>0.
\end{equation}
We continue with the calculation of $K(1-\alpha;x)$
\begin{align*}
	K(1-\alpha;x)&=	\int_{0}^x \left(  \frac{1}{B(\beta_1,\beta_2)}t^{\beta_1 -1}(1-t)^{\beta_2 -1} \right)^{1-\alpha} \left( \frac{1}{B(\gamma_1,\gamma_2)}t^{\gamma_1 -1}(1-t)^{\gamma_2 -1}  \right)^{\alpha}\, dt\\[0.5em]
	& = \frac{1}{B(\beta_1,\beta_2)}\left[ \frac{ B(\beta_1,\beta_2)}{ B(\gamma_1,\gamma_2)}\right]^{\alpha} \int_{0}^{x} t^{\delta_1(\alpha)-1}(1-t)^{\delta_2(\alpha)-1}\, dt \\[0.5em]
	& = \frac{1}{B(\beta_1,\beta_2)}\left[ \frac{ B(\beta_1,\beta_2)}{ B(\gamma_1,\gamma_2)}\right]^{\alpha}  B(x;\delta_1(\alpha),\delta_2(\alpha))\,.
\end{align*}
We substitute the outcome of the previous calculations in \eqref{eq:y(x)0} and obtain
$$
\frac{B(y;\gamma_1,\gamma_2)}{B(\gamma_1,\gamma_2)} = \frac{\frac{1}{B(\beta_1,\beta_2)}\left[ \frac{ B(\beta_1,\beta_2)}{ B(\gamma_1,\gamma_2)}\right]^{\alpha}  B(x;\delta_1(\alpha),\delta_2(\alpha)) }{\frac{1}{B(\beta_1,\beta_2)}\left[ \frac{ B(\beta_1,\beta_2)}{ B(\gamma_1,\gamma_2)}\right]^{\alpha} B(\delta_1(\alpha),\delta_2(\alpha)) } = \frac{B(x;\delta_1(\alpha),\delta_2(\alpha)) }{ B(\delta_1(\alpha),\delta_2(\alpha))},
$$
which gives
\begin{equation}\label{eq:betaxy}
	x(y) = B^{-1}\left(\frac{B(\delta_1(\alpha),\delta_2(\alpha)) }{B(\gamma_1,\gamma_2) }B(y;\gamma_1,\gamma_2); \delta_1(\alpha),\delta_2(\alpha)  \right).
\end{equation}
We end up the deduction of the transformed density by inserting \eqref{eq:betaxy} into \eqref{eq:tilderho} to find
\begin{equation*}
	\begin{split}
		\Appr{\pdfr}{\alpha}(y)&=K_\alpha[\pdfr||\pdf]\left(\frac{\pdf(x(y))}{\pdfr(x(y))}\right)^\alpha\,\pdfr(y)\\[0.5em]
		&= \frac{1}{B(\beta_1,\beta_2)}\left[ \frac{ B(\beta_1,\beta_2)}{ B(\gamma_1,\gamma_2)}\right]^{\alpha} B(\delta_1(\alpha),\delta_2(\alpha))\\
		&\times\left[\frac{\frac{1}{B(\beta_1,\beta_2)}x(y)^{\beta_1 -1}(1-x(y))^{\beta_2 -1} }{\frac{1}{B(\gamma_1,\gamma_2)}x(y)^{\gamma_1 -1}(1-x(y))^{\gamma_2 -1} } \right]^{\alpha} \frac{1}{B(\gamma_1,\gamma_2)}y^{\gamma_1 -1}(1-y)^{\gamma_2 -1}\\[0.5em]
		& =\frac{B(\delta_1(\alpha),\delta_2(\alpha)) }{B(\beta_1,\beta_2) B(\gamma_1,\gamma_2)}x(y)^{\alpha(\beta_1-\gamma_1)}(1-x(y))^{\alpha(\beta_2 - \gamma_2)}y^{\gamma_1 -1}(1-y)^{\gamma_2 -1},
	\end{split}
\end{equation*}
where $x(y)$ is given in \eqref{eq:betaxy}. We close this section by observing that the critical values of $\alpha$ defined in \eqref{eq:asup} are in this case derived from the conditions in \eqref{eq:interm16}, depending on the signs of $\gamma_1-\beta_1$ and $\gamma_2-\beta_2$.

\section{The transformation of a $N-$piecewise density}\label{sec:Npiece}

This section is devoted to the general computation of the divergence-transformed density of a $N$-piecewise probability density function. As we shall see afterwards in Section \ref{sec:numerical}, this calculation has interesting applications, providing a new approximation technique of the (arbitrary) reference density function $h$ by simple functions, according to the outcome of Theorem \ref{th:divergence}. Let us consider thus a fixed reference density function $h$ and a $N$-piecewise probability density function $f_N$, both defined on the same support $\Omega=[x_0,x_N]\subseteq \mathbb{R}$ and satisfying \eqref{cond:support}. We write $f_N$ as a sum of its constant parts:
\begin{equation}\label{eq:Npiecewise}
f_N(x) = \sum_{j=1}^N C_j  \,\chi_{\Delta_j}(x), \qquad \bigcup_{j=1}^N \Delta_j = \Omega, \quad \Delta_j \cap \Delta_k = \emptyset , \,\,\, \forall j \neq k.
\end{equation}
We next compute, for $\alpha\in\Rset$, the $\alpha$-transformed density of $f_N$ according to Eq. \eqref{eq:tilderho}:
\begin{align}\label{eq:fAtrans}
(f_N)^{(\mathfrak{A},h)}_{\alpha}(y) &= K_\alpha[\pdfr||f_N]\left[\frac{f_N(x(y))}{h(x(y))} \right]^{\alpha}h(y)  = K_\alpha[\pdfr||f_N]\left[\frac{ \sum_{j=1}^N C_j \,\chi_{\Delta_j}(x(y)) }{h(x(y))} \right]^{\alpha}h(y) \nonumber   \\
& =K_\alpha[\pdfr||f_N]
 \sum_{j=1}^{N} \left[\frac{C_j}{h(x(y))}\right]^\alpha \, h(y) \, \chi_{\Delta_j}(x(y)) \nonumber  \\
& =K_\alpha[\pdfr||f_N]
\sum_{j=1}^{N} \left[\frac{C_j}{h(x(y))}\right]^\alpha \, h(y) \, \chi_{\widetilde{\Delta}_j}(y)\, ,
\end{align}
where $\widetilde{\Delta}_j=y(\Delta_j),$ and
\begin{equation}
	K_\alpha[\pdfr||f_N]=	\left(  \sum_{j=1}^{N} C_j^{1-\alpha}\int_{\Delta_j} h^{\alpha}(t)\, dt \right).
\end{equation}
We recall that the change of variable is given by
\begin{equation*}
H(y) = \int_{x_i}^{y} h(s)\, ds=\frac{K(1-\alpha;x)}{K_\alpha[\pdfr||f_N]}\, ,
\end{equation*}
where, for any $m\leqslant N$ and $x\in\Delta_m:=[x_{m-1},x_m)$,
\begin{equation*}
\begin{split}
K(1-\alpha;x) & = \int_{x_0}^{x} h^{\alpha}(t)f_N^{1-\alpha}(t)\, dt\\
&= \int_{x_0}^{x_{m-1}} h^{\alpha}(t)f_N^{1-\alpha}(t)\, dt  + \int_{x_{m-1}}^{x} h^{\alpha}(t)f_N^{1-\alpha}(t)\, dt
\\
&= \sum_{i=1}^{m-1} C_{i}^{1-\alpha}\int_{\Delta_{i}} h^{\alpha}(t)\, dt  + C_{m}^{1-\alpha} \int_{x_{m-1}}^{x} h^{\alpha}(t)\, dt.
\end{split}
\end{equation*}
We then continue the previous calculations to obtain that
	\begin{align*}
	H(y) &= \frac{\mathcal J_\alpha^{(m)}[h||f_N]+ C_{m}^{1-\alpha} \int_{x_{m-1}}^{x} h^{\alpha}(t)\, dt}{K_\alpha[h|| f_N]},
\end{align*}
with
\begin{equation*}
\mathcal J_\alpha^{(m)}[h||f_N]=\sum_{i=1}^{m-1} C_{i}^{1-\alpha}\int_{\Delta_{i}} h^{\alpha}(t)\, dt
\end{equation*}
for any $m\leq N$ and $x\in\Delta_m=[x_{m-1},x_m]$. Consequently,
\begin{align*}
	\int_{x_{m-1}}^{x} h^{\alpha}(t)\, dt & =C_m^{\alpha-1}\left(\mathcal K_\alpha[h||f_N]\,	H(y) - \mathcal J_\alpha^{(m)}[h||f_N]\right),
\end{align*}
and the latter equality allows to undo the change of variable and calculate the inverse $x=x(y)$.

We next compute the Kullback-Leibler divergence of the $\alpha$-transformed density of $f_N$ and $h$ as an application of Eq.~\eqref{eqlemma:appr_diver} for $\xi=1$:
\begin{align*}
D\left[(f_N)^{(\mathfrak{A},h)}_{\alpha}|| h\right]
& = \alpha D[f_N || h]\, +\log(K_\alpha[\pdfr||f_N]).
\end{align*}
We therefore have to calculate the Kullback-Leibler divergence of $f_N$ and $h$, which gives
\begin{align*}
		D\left[f_{N}|| h\right] &= \int_{\Omega} f_{N}(x) \log\left[\frac{f_{N}(x)}{h(x)}\right] \, dx  =\sum_{i=1}^N \int_{\Delta_i} f_{N}(x) \log\left[\frac{f_{N}(x)}{h(x)}\right] \, dx\\
		& =  \sum_{i=1}^N C_i  \int_{\Delta_i} \log\left[\frac{C_i }{h(x) }  \right]\, dx\,.
\end{align*}
We thus deduce that
\begin{equation}\label{eq:DKLfalpha}
D\left[(f_N)^{(\mathfrak{A},h)}_{\alpha}|| h\right] = \alpha \sum_{i=1}^N C_i  \int_{\Delta_i} \log\left[\frac{C_i }{h(x) }  \right]\, dx+ \log\left[\sum_{i=1}^{N} C_i^{1-\alpha}\int_{\Delta_i} h^{\alpha}(x) \, dx \right]\,.
\end{equation}
Note that in the case $\alpha=0$ ones obtains $D[(f_N)^{(\mathfrak{A},h)}_{\alpha}||h]=0,$ as expected.

\section{Numerical examples: approaching and separating}\label{sec:numerical}

As an application of the theoretical examples calculated in the previous two sections, we illustrate in this section by means of numerical experiments how a pdf either approach or separate from certain reference pdf as the transformation parameter varies. In the first experiment we show how a transformed piecewise function, $f^{(\mathfrak{A},h)}_{N,\alpha}(y)$, continuously approach a reference pdf, $h(x)$, as the parameter of the transformation $\alpha$ decreases to $0$. For this matter, we choose as reference pdf the truncated normal distribution,
\[
h(x;\mu,\sigma,a,b) = \begin{cases} \frac{1}{\sigma}\frac{\varphi\left(\frac{x-\mu}{\sigma}\right) }{\Phi\left(\frac{b - \mu }{\sigma }  \right) - \Phi\left(\frac{a - \mu }{\sigma }  \right)}, & [a,b] = [-5,5] \\
	0, & \text{otherwise}\,
\end{cases}
\]
for $\mu = 0$ and $\sigma = 1$, where
\[
\varphi(x) = \frac{1}{\sqrt{2\pi}}e^{-\frac{x^2}{2}}, \quad \Phi(x) = \frac{1}{2}\left[1+\text{erf}\left(\frac{x}{\sqrt{2}}\right) \right]\, .
\]
The piecewise function to which we apply the transformation is given by
\[
f_4(x) =  \sum_{j=1}^{4} C_j \, \chi_{\Delta_j},
\]
with $\{C_j\}_{j=1}^{4} = \{\frac{1}{3},\frac{1}{25},\frac{1}{50},\frac{344}{4275} \}$, $\Delta_1 = \left[-5,-\frac{63}{20}\right]$, $\Delta_2 = \left[-\frac{63}{20}, -\frac{3}{4}\right]$, $\Delta_3 = \left[-\frac{3}{4}, \frac{43}{20}\right]$, $\Delta_4 = \left[\frac{43}{20}, 5\right]$. As stated in Eq.~\eqref{eq:fAtrans}, the transformed piecewise function reads
\[
(f_4)^{(\mathfrak{A},h)}_{\alpha}(y) = K_\alpha[h||f_4]\sum_{j=1}^{4} \left[\frac{C_j}{h(x(y))}\right]^\alpha \, h(y,0,1,-5,5) \, \chi_{\tilde{\Delta}_j}(y)\, .
\]
We show in Figures~\eqref{fig:experiment1_fig_1}-\eqref{fig:experiment1_fig_8} how the transformed piecewise pdf approaches $h(x)$ in a smooth way as $\alpha \to 0$. In this process, one can observe that the divergence transformation applied to the piecewise function modifies both its shape and the partition of its domain, in order to approach $h(x)$, while preserving the probability in each interval.

\begin{figure}[H]
	\centering
	\medskip
	
	\begin{subfigure}[t]{.3\linewidth}
		\centering\includegraphics[width=.5\linewidth]{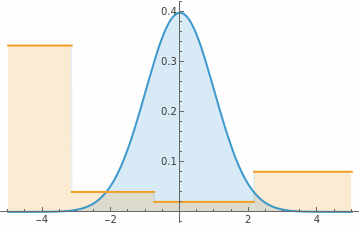}
		\caption{$\alpha =1$}
			\label{fig:experiment1_fig_1}
	\end{subfigure}
	\begin{subfigure}[t]{.3\linewidth}
		\centering\includegraphics[width=.5\linewidth]{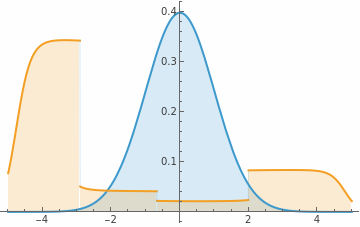}
		\caption{$\alpha =\frac{9}{10}$}
		\label{fig:experiment1_fig_2}
	\end{subfigure}
	\begin{subfigure}[t]{.3\linewidth}
		\centering\includegraphics[width=.5\linewidth]{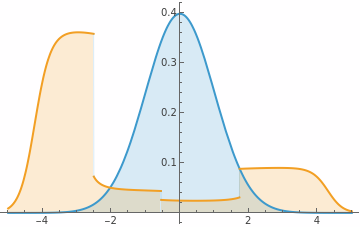}
		\caption{$\alpha =\frac{3}{4}$}
		\label{fig:experiment1_fig_3}
	\end{subfigure}
	
	\bigskip\hrulefill\bigskip
	
	\medskip
	
	\begin{subfigure}{.3\linewidth}
		\centering\includegraphics[width=.5\linewidth]{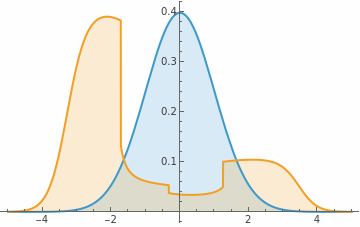}
		\caption{$\alpha =\frac{1}{2}$}
		\label{fig:experiment1_fig_4}
	\end{subfigure}
	\begin{subfigure}{.3\linewidth}
		\centering\includegraphics[width=.5\linewidth]{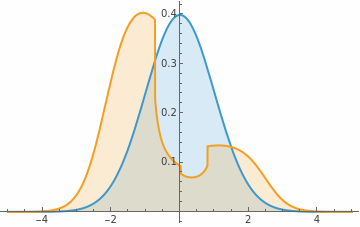}
		\caption{$\alpha =\frac{1}{4}$}
		\label{fig:experiment1_fig_5}
	\end{subfigure}
	\begin{subfigure}{.3\linewidth}
		\centering\includegraphics[width=.5\linewidth]{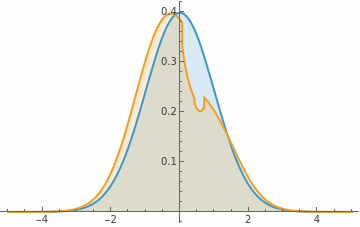}
		\caption{$\alpha =\frac{1}{16}$}
		\label{fig:experiment1_fig_6}
	\end{subfigure}
	
	\bigskip\hrulefill\bigskip
	
	\medskip
	
	\begin{subfigure}[b]{.3\linewidth}
		\centering\includegraphics[width=.5\linewidth]{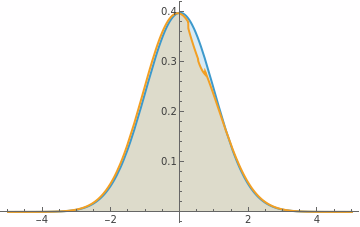}
		\caption{$\alpha =\frac{1}{64}$}
		\label{fig:experiment1_fig_7}
	\end{subfigure}
	\begin{subfigure}[b]{.3\linewidth}
		\centering\includegraphics[width=.5\linewidth]{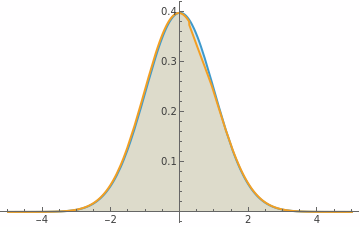}
		\caption{$\alpha =\frac{1}{128}$}
		\label{fig:experiment1_fig_8}
	\end{subfigure}
	\label{fig:experiment1}
	\caption{Sequence of approach of the transformed $4$-piecewise pdf $(f_4)^{(\mathfrak{A},h)}_{\alpha}(y)$ (orange) to the reference pdf $h(x)$ (blue) as $\alpha\to 0$.}
\end{figure}

Note that, in the previous experiment, the probability $p_1$ of the first region of the support $\Delta_1$ of $f_4$ is approximately $p_1\simeq 0.61>1/2$. In addition, in the limit $\alpha=0$, Eq.~\eqref{eq:y(x)} gives $H(y)=F(x)$, where $F(x)$ is the cumulative function of the density $f_4$. It follows that, as $\alpha\to0$, the transformed region $y(\Delta_1)$ is necessarily the one whose probability is also $p_1$ for the proper reference function $h$ (since the probability is preserved for any $\alpha\in\mathbb R$), which also implies a transfer of mass during the transformation, as the final density is symmetric. In addition, regarding the regions $\Delta_2$ and $\Delta_3$ of the initial density $f_4$ (which have a lower mass than $\Delta_1$), we may observe that the length of the transformed support when $\alpha\to0$ decreases according to the corresponding probability in the support of $h$.

In the next example, we employ the same transformation to increase the dissimilarity between the transformed pdf and a reference one in order to achieve a better discrimination between them. This aspect is crucial in those
processes in which the goal is to distinguish a singular event from a set of events labeled as \textit{normal}. For example, this is essential in the methodologies designed to detect one or more epileptic seizures from the temporal signal associated with an electroencephalogram which, due to its own nature, tends to be very noisy \cite{Squicciarini2024a}. In order to illustrate this fact, we consider the reference pdf,
\[
h(x) =  \frac{1}{\sqrt{2\pi}}e^{-\frac{x^2}{2}},
\]
and the so-called $\epsilon$-skew-normal distribution~\cite{Mudholkar2000} defined as
\[
f_{\epsilon}(x) = N_{\epsilon} \begin{cases}
			\frac{1}{\sqrt{2\pi }}e^{-\frac{x^2}{2(1+\epsilon)^2} },& x< 0 \\
			\frac{1}{\sqrt{2\pi }}e^{-\frac{x^2}{2(1-\epsilon)^2} },& x \geq 0 \, ,\\			
			\end{cases}
\]
where $N_{\epsilon} = \frac{2}{\epsilon + 1 + |\epsilon - 1|}$ denotes the normalization constant, with $\epsilon>0$. In this experiment we choose $\epsilon = \frac{1}{5}$.

\begin{figure}[H]
	\centering
	\medskip
	
	\begin{subfigure}[t]{.3\linewidth}
		\centering\includegraphics[width=.5\linewidth]{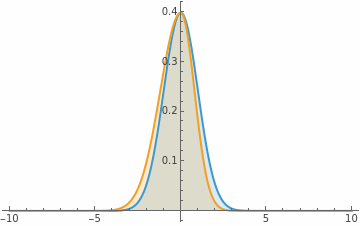}
		\caption{$\alpha =1$}
			\label{fig:experiment2_fig_1}
	\end{subfigure}
	\begin{subfigure}[t]{.3\linewidth}
		\centering\includegraphics[width=.5\linewidth]{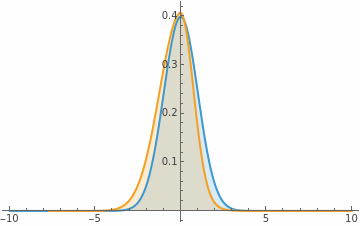}
		\caption{$\alpha =1.25$}
		\label{fig:experiment2_fig_2}
	\end{subfigure}
	\begin{subfigure}[t]{.3\linewidth}
		\centering\includegraphics[width=.5\linewidth]{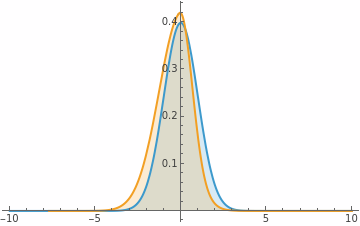}
		\caption{$\alpha =1.5$}
		\label{fig:experiment2_fig_3}
	\end{subfigure}
	
	\bigskip\hrulefill\bigskip
	
	\medskip
	
	\begin{subfigure}{.3\linewidth}
		\centering\includegraphics[width=.5\linewidth]{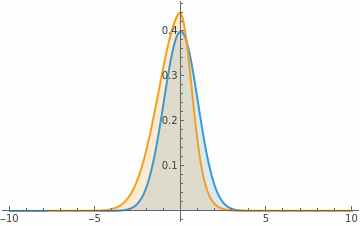}
		\caption{$\alpha =1.75$}
		\label{fig:experiment2_fig_4}
	\end{subfigure}
	\begin{subfigure}{.3\linewidth}
		\centering\includegraphics[width=.5\linewidth]{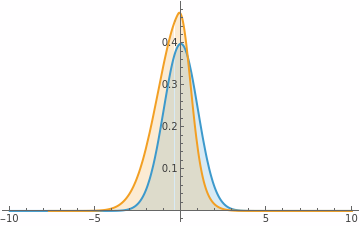}
		\caption{$\alpha =2$}
		\label{fig:experiment2_fig_5}
	\end{subfigure}
	\begin{subfigure}{.3\linewidth}
		\centering\includegraphics[width=.5\linewidth]{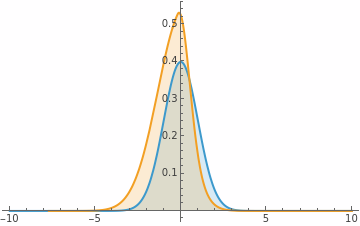}
		\caption{$\alpha =2.25$}
		\label{fig:experiment2_fig_6}
	\end{subfigure}
	
	\bigskip\hrulefill\bigskip
	
	\medskip
	
	\begin{subfigure}[b]{.3\linewidth}
		\centering\includegraphics[width=.5\linewidth]{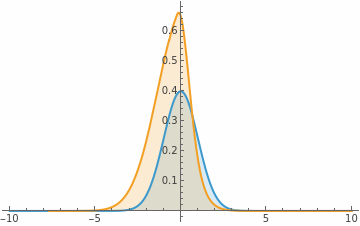}
		\caption{$\alpha =2.5$}
		\label{fig:experiment2_fig_7}
	\end{subfigure}
	\begin{subfigure}[b]{.3\linewidth}
		\centering\includegraphics[width=.5\linewidth]{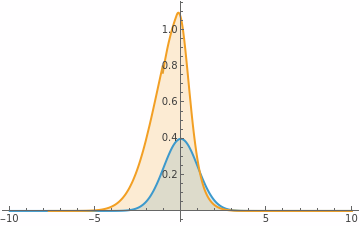}
		\caption{$\alpha =2.7$}
		\label{fig:experiment2_fig_8}
	\end{subfigure}
	\label{fig:experiment2}
	\caption{Sequence of separation of the transformed pdf $(f_{\epsilon})^{(\mathfrak{A},h)}_{\alpha}(y)$ with respect to the reference pdf $h(x)$ as $\alpha\to \alpha_c^{+}$.}
\end{figure}

From Figures \eqref{fig:experiment2_fig_1}-\eqref{fig:experiment2_fig_8}, corresponding to the second experiment, we observe that, as $\alpha$ increases, the mass of the transformed pdf $(f_{\epsilon})^{(\mathfrak{A},h)}_{\alpha}(y)$ seems to move to the region $y\in (-\infty,0)$. Similarly to what we have seen in the former experiment, a mass transfer from the half-line with bigger mass of the initial density $f_{\epsilon}$ (that is, the region $x\in(-\infty,0)$) to the complementary half-line occurs as $\alpha\to0$ since the reference function $h$ is symmetric. Recalling the group structure of the divergence transformation (and in particular Eq. \eqref{eq:inverse_appr}), we deduce that, in the opposite direction (that is, $\alpha>1$) the contrary phenomena is expected to occur. This is seen in Figures \eqref{fig:experiment2_fig_1}-\eqref{fig:experiment2_fig_8}, since the original assymmetry is increased, generating an increment of the mass in the negative half-line. Moreover, it also appears that the maximum of the transformed density slowly moves to the left.

Another remarkable fact is the increase in slope in each region. This pattern seen in the figures is explained by similar calculations as the ones performed in Section~\ref{subsec:twogauss}. Note that the variance of the Gaussian considered for $x>0$ in the definition of the original pdf $f_{\epsilon}$ is $1-\epsilon<1$, while the variance of the Gaussian considered for $x<0$ in the definition of $f_{\epsilon}$ is $1+\epsilon>1$. Thus, taking separately the two parts, the critical value $\alpha_c^+$ corresponding to the variance $1-\epsilon$ is finite, while the critical value $\alpha_c^+$ corresponding to the variance $1+\epsilon$ is $+\infty$, according to a similar calculation as the one in Eq.~\eqref{eq:acrit}. Thus, when $\alpha$ increases, it will approach at some point the finite value of $\alpha_c^{+}$ corresponding to the variance $1-\epsilon$, which applies to the part of the transformed density arriving from the part of the original density $f_{\epsilon}$ lying in the positive half-line. Consequently, when $\alpha\to\alpha_c^+< \infty$, the part of the total mass of the original density $f_{\epsilon}$ lying in the region $x\geq 0$ concentrates to a Dirac distribution, while in the meantime the part of the total mass corresponding to complementary region is not yet concentrated, since for the variance of the Gaussian applied in the region $x<0$, $\alpha_c^{+} = \infty$.

We inspect the transformed pdf for $\alpha=2.25$ in a closer form in the next figure by plotting separately the deformed densities $(f_\epsilon)_\alpha^{(\mathfrak A, h)}(y(x))$ corresponding to the restrictions of the original density $f_{\epsilon}$ to the half-lines $x<0$, respectively $x\geqslant0$.
\begin{figure}[H]
	\centering
	\medskip
	
	\begin{subfigure}[t]{.45\linewidth}
		\centering\includegraphics[width=.5\linewidth]{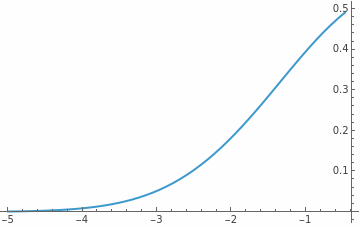}
		\caption{$(f_\epsilon)_\alpha^{(\mathfrak A, h)}(y)$ for $y\in(-5,y(0))$}
		\label{fig:experiment3_fig_31}
	\end{subfigure}
	\begin{subfigure}[t]{.45\linewidth}
		\centering\includegraphics[width=.5\linewidth]{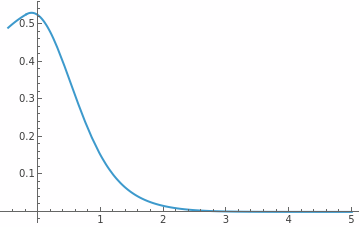}
		\caption{$(f_\epsilon)_\alpha^{(\mathfrak A, h)}(y)$ for $y\in(y(0),5)$}
		\label{fig:experiment3_fig_32}
	\end{subfigure}
\end{figure}
We observe that the transformed pdf has a single maximum point but, interestingly, this maximum is achieved neither at $y(0)$, nor at $y=0$, but at some intermediate point between these two values. This is due to the fact that, as previously discussed, the region with $y>y(0)$ tends to a Dirac distribution concentrated at $y(0)$ as $\alpha\to\alpha_c^+$, but in the meantime the function $(f_\epsilon)_\alpha^{(\mathfrak A, h)}$ preserves its continuity at $y(0)$, thus the maximum point moved from $x=0$ to the negative side, but did not reach yet the limit value $y(0)$.

\section{Conclusions}\label{sec:conclusions}

In this work we have introduced a family of transformations allowing to interpolate between any pair of probability density functions sharing a common support and being strictly positive on that support, in such a way the family of Kullback-Leibler and Rényi divergences changes in a monotone way. We have then proved that this family of transformations, which we have named \textit{divergence transformations}, inherits a group structure from the fact that they are the algebraic conjugation of two transformations recently introduced and related to the Sundman transformations of power type: the differential-escort transformation and its relative counterpart. We show that in some limit cases of the parameter the transformations are no longer well-defined and the transformed pdf tends towards a Dirac distribution. In addition, we also propose the definition of some complexity measures of LMC-type involving divergences and some (recently defined) relative measures of generalized moment and Fisher types. We establish the monotonicity properties of these statistical measures with respect to the aforementioned family of transformations or suitable algebraic conjugations of them with the recently introduced up and down transformations.

Several examples of interest from both theoretical and numerical perspectives are given, involving the exponential, Gaussian, asymmetric Gaussian, Beta and $N$-piecewise functions. In the latter case, we show how any simple function can be arbitrarily approximated to a reference function $h$ by applying the divergence transformation. In contrast, the example with the Gaussian and a particular case of an asymmetric Gaussian shows how the same transformation can be used to separate pdfs that are arbitrarily close to each other. Interesting phenomena are discovered and described alongside these processes. We highlight again that, in both cases, the whole family of Rényi divergences and their complexity measure counterpart vary in a monotone way with respect to the transformation parameter for any pair of probability densties.

We believe that many future applications in information theory and other fields of science and engineering could be derived by applying the divergence transformations with respect to suitably chosen reference functions $h$, according to each case of interest. Further applications can be discovered by employing the relative versions of the complexity measures established in this paper in order to gain insight into the structure of diverse systems. Thus, we consider that this family of divergence transformations could be a strong tool for further investigation.

\subsection*{Acknowledgements}

R. G. I. is partially supported by the project PID2024-160967NB-I00 (AEI) funded by the Ministry of Science, Innovation and Universities of Spain. D. P.-C. and E. V. T. are partially supported by the project PID2023-153035NB-100 (AEI) funded by the Ministry of Science, Innovation and Universities of Spain and “ERDF/EU A way of making Europe”. D. P.-C. also wants to thank to the Professor A. Zarzo by inspiring discussions at the first stages of this research work.
\bigskip

\noindent \textbf{Data availability} Our manuscript has no associated data.

\bigskip

\noindent \textbf{Competing interest} The authors declare that there is no competing interest.

\bibliographystyle{unsrt}
\bibliography{refs}
\end{document}